\documentclass[twocolumn,aps,pra,notitlepage,superscriptaddress,longbibliography]{revtex4-2}
\renewcommand{\selectlanguage}[1]{}
\usepackage[breaklinks=true]{hyperref}
\hypersetup{
    colorlinks,
    linkcolor={red!50!black},
    citecolor={blue!50!black},
    urlcolor={blue!80!black}
}
\usepackage{amsmath,amssymb,amsthm,mathtools,mathrsfs}
\usepackage{graphicx}
\usepackage[table,dvipsnames]{xcolor}
\usepackage{tikz}
\usetikzlibrary{math,calc}
\usetikzlibrary{arrows.meta}
\usetikzlibrary{decorations.pathmorphing,shapes}
\usepackage{quantikz}
\usepackage[all]{xy} 
\usepackage{enumerate}
\usepackage[normalem]{ulem}
\usepackage{graphicx}
\usepackage{epsfig}
\usepackage{pgfplots}
\pgfplotsset{compat=1.14}
\usepackage{float}
\usepackage{adjustbox}
\usepackage{calligra}
\usepackage{dsfont}

\newcounter{sarrow}

\theoremstyle{plain}
\newtheorem{theorem}[equation]{Theorem}
\newtheorem{lemma}[equation]{Lemma}
\newtheorem{proposition}[equation]{Proposition}
\newtheorem{corollary}[equation]{Corollary}
\theoremstyle{definition}
\newtheorem{definition}[equation]{Definition}

\newtheorem{example}[equation]{Example} 

\newcommand{\id}{\mathrm{id}}
\newcommand{\matr}{\mathbb{M}}
\newcommand{\Tr}{\operatorname{Tr}}

\DeclareMathOperator{\SWAP}{\tt SWAP} 
\DeclareMathOperator{\Ad}{\mathrm{Ad}} 
\DeclareFontFamily{OT1}{pzc}{}
\DeclareFontShape{OT1}{pzc}{m}{it}{ <-> s*[1.2] pzcmi7t }{}
\DeclareMathAlphabet{\mathpzc}{OT1}{pzc}{m}{it}

\DeclareMathAlphabet{\mathcalligra}{T1}{calligra}{m}{n}
\DeclareFontShape{T1}{calligra}{m}{n}{<->s*[2.2]callig15}{}
\newcommand{\scripty}[1]{\ensuremath{\mathcalligra{#1}}}

\newcommand{\states}{{\!\scriptscriptstyle\scripty{S}}\;\,}
\newcommand{\unitary}{{\!\scriptscriptstyle\scripty{U}}\;\,}

\def\R{{{\mathbb R}}}
\def\C{{{\mathbb C}}}

\def\CP{{{\mathbb C}{\mathbb P}}}

\def\N{{{\mathbb N}}}
\def\Z{{{\mathbb Z}}}

\newcommand\define[1]{\emph{\textbf{#1}}}

\pgfdeclareradialshading[tikz@ball]{ball}{\pgfqpoint{0bp}{0bp}}{%
 color(0bp)=(tikz@ball!0!white);
 color(7bp)=(tikz@ball!0!white);
 color(15bp)=(tikz@ball!70!black);
 color(20bp)=(black!70);
 color(30bp)=(black!70)}
\makeatother

\tikzset{viewport/.style 2 args={
    x={({cos(-#1)*1cm},{sin(-#1)*sin(#2)*1cm})},
    y={({-sin(-#1)*1cm},{cos(-#1)*sin(#2)*1cm})},
    z={(0,{cos(#2)*1cm})}
}}

\pgfplotsset{only foreground/.style={
    restrict expr to domain={rawx*\CameraX + rawy*\CameraY + rawz*\CameraZ}{-0.05:100},
}}
\pgfplotsset{only background/.style={
    restrict expr to domain={rawx*\CameraX + rawy*\CameraY + rawz*\CameraZ}{-100:0.05}
}}

\def\addFGBGplot[#1]#2;{
    \addplot3[#1,only background, opacity=0.30] #2;
    \addplot3[#1,only foreground] #2;
}

\begin{document}

\title{Towards structure-preserving quantum encodings}

\author{Arthur J.\ Parzygnat}
\email{arthurjp@mit.edu}
\affiliation{Experimental Study Group, Massachusetts Institute of Technology, Cambridge, Massachusetts 02139, USA}
\affiliation{
Deloitte Consulting LLP
}

\author{Tai-Danae Bradley}
\email{tai.danae@math3ma.com}
\affiliation{
SandboxAQ, Palo Alto, California 94301, USA
}
\affiliation{Department of Mathematics, The Master's University, Santa Clarita, California 91321, USA}

\author{Andrew Vlasic}
\email{avlasic@deloitte.com}
\affiliation{
Deloitte Consulting LLP
}

\author{Anh Pham}
\email{anhdpham@deloitte.com}
\affiliation{
Deloitte Consulting LLP
}

\date{\today}

\begin{abstract}
Harnessing the potential computational advantage of quantum computers for machine learning tasks relies on the uploading of classical data onto quantum computers through what are commonly referred to as quantum encodings. The choice of such encodings may vary substantially from one task to another, and there exist only a few cases where structure has provided insight into their design and implementation, such as symmetry in geometric quantum learning. Here, we propose the perspective that category theory offers a natural mathematical framework for analyzing encodings that respect structure inherent in datasets and learning tasks. We illustrate this with pedagogical examples, which include geometric quantum machine learning, quantum metric learning, topological data analysis, and more. Moreover, our perspective provides a language in which to ask meaningful and mathematically precise questions for the design of quantum encodings and circuits for quantum machine learning tasks. 
\end{abstract}


\maketitle

\section{Introduction}

A critical step in quantum machine learning (QML) for classical data is deciding how to encode the data in a Hilbert space through a quantum circuit~\cite{Aaronson15,sim2019expressibility,huang2021power,thanasilp2024exponential,bowles2024,KhAmSi24}. 
Finding the best encoding for a specific task within a large space of options is often a time- and resource-intensive task, and rarely are there any general-purpose guiding principles to choose encodings.  
Here, we utilize category theory to organize and isolate structure in the dataset and learning task~\cite{mac2013categories}. This restricts the initial unstructured space of all quantum encodings to a structured subspace that is usually significantly smaller, thus enabling the reduction of resources needed to find suitable encodings.
By isolating the structure that quantum encodings preserve for a given learning task, such as a symmetry, a metric, or a topology, one identifies the encodings as the space of structure-preserving morphisms in the appropriate category, thereby providing a mathematical model for what constitutes structure preservation.

It is only after encoding the classical data onto a quantum computer that one can try to take advantage of potential speedups due to quantum information processing~\cite{NiCh11,BCRS19}. The specifics of what happens after encoding depends on the model~\cite{Jerbi2023,schuld21}. For example, in variational QML models, embedded data is processed through a parameterized quantum circuit with tunable parameters before obtaining an output through measurement. The measurements can then be used to provide a feedback loop to update the parameters defining the parameterized quantum circuit in order to improve the algorithm for unseen data~\cite{Benedetti2019}. Finding optimal parameters is often plagued by barren plateaus, and much recent work is spent on understanding and mitigating the effects of such barren plateaus in variational algorithms~\cite{mcclean2018barren,larocca2024reviewbarrenplateausvariational,Ragone2024LieAlgebraBarrenPlateaus,Fontana_2024_barrenplateaus,cerezo2024doesprovableabsencebarren}. 

But it is the first step of this procedure, the embedding strategy --- sometimes called \emph{state preparation} \cite{araujo2021divide,schuld2018circuitcentric,Tang2021} or \emph{data encoding} \cite{araujo2021divide, schuld2021effect, schuld21, larose2020robust} --- which may arguably play the most significant role. More concretely, it involves the choice of a function, often referred to as a \emph{(quantum) feature map} \cite{lloyd20, schuld2019quantum, havlivcek2019supervised, huang2021power, goto2020universal, schuld2018circuitcentric,KiBe22}, that assigns to each data point a quantum state, known as a \emph{feature vector} \cite{schuld2019quantum}, \emph{data encoding} \cite{schuld21, huang2021power, larose2020robust}, or \emph{(quantum) embedding} \cite{lloyd20, huang2021power, thanasilp2023subtleties, KiBe22}, to name a few. In this article, we use \emph{quantum encoding} to refer to the function, but the terminology varies as do the encoding choices, each of which introduces different inductive biases \cite{cerezo2022challenges}. Notably, there is no general framework for deciding what constitutes a good quantum encoding given a particular dataset or learning task. The notion of ``good'' is indeed multifaceted and may encompass several (sometimes competing) aspects required for the success of a QML model. For example, previous works have focused on quantum encodings' expressivity \cite{schuld2021effect,goto2020universal}, expressibility~\cite{sim2019expressibility}, robustness to noise \cite{larose2020robust}, computational cost \cite{araujo2021divide}, and effect on model trainability \cite{thanasilp2023subtleties}. 

Another facet, one that we focus on in this paper, is the interplay between quantum encodings and structure in data. One archetypal example is the field of geometric quantum machine learning, where group theoretic tools aid in constructing quantum encodings that respect a given symmetry~\cite{meyer2023exploiting,ragone2023representation,Nguyen_equivariantQNN2024}. This field stems from its successful predecessor of (classical) geometric learning~\cite{fukushima1980neocognitron,bronstein2021GDL}, whose goal is to characterize and analyze the structure of data from the perspective of symmetry~\cite{BBLSV17}. Generalizing beyond symmetry, how should a quantum encoding be constructed so that a particular structure is preserved? And what is an appropriate definition of ``structure'' to use in this context? 

Similar questions have been asked by Bowles, Ahmed, and Schuld, who recently called for a more scientifically rigorous approach to benchmarking in QML \cite{bowles2024}:
\begin{quote}
More studies that focus on questions of structure in data are crucial for the design of meaningful benchmarks: What mathematical properties do real-world applications of relevance have? ... How can we connect them to the mathematical properties of quantum models?
\end{quote}

In a similar vein, Larocca et al.\ stated ~\cite{larocca2022group}: 

\begin{quote}
...models with little to no inductive biases (i.e., with no assumptions about the problem embedded in the model) are likely to have trainability and generalization issues especially for large problem sizes. As such, it is fundamental to develop schemes that encode as much information as available about the problem at hand.
\end{quote}

As another example, Thanasilp, Wang, Cerezo, and Holmes pointed out that many of the difficulties of extracting information from data transferred to quantum states are due to the exponential concentration of quantum kernels. This is caused by various factors, one of which includes the expressivity of quantum encodings~\cite{thanasilp2024exponential}:
\begin{quote}
Our work on embedding-induced concentration suggests that problem-inspired embeddings should be used over problem-agnostic embeddings (which are typically highly expressive and entangling). 
\end{quote}
\begin{quote}
...Unstructured data embeddings should generally be avoided and the data structure should be taken into account when designing a data-embedding....
\end{quote}

The \emph{problem-inspired embeddings} that are mentioned here are precisely those types of encodings that preserve a certain structure. Such encodings form the focus of this paper. 

Towards addressing questions of structure in data, a quick perusal through the literature shows that a number of mathematical structures have relevancy depending on the dataset and task at hand. Geometric quantum machine learning, as mentioned above, exploits symmetry in data, whereas quantum metric learning prioritizes similarity and distance structure~\cite{lloyd20}. On the other hand, quantum topological data analysis~\cite{LloydGarneroneZanardi16}, which has been studied in the context of quantum encodings, concerns topological structure in data~\cite{vlasic2023qtda}. Each involves the analysis of a different mathematical structure and the extent to which it is or is not preserved under a passage to a Hilbert space under a nonlinear mapping. A methodical study of quantum encodings can therefore benefit from a principled framework in which to think about structured mathematical objects and structure-preserving mappings between them. 

Category theory, a relatively modern branch of mathematics, is the study of precisely this~\cite{riehl,mac2013categories,Pe19,Leinster2014basic} (see also Refs.~\cite{HeVi19,CoKi17} for quantum-focused introductions). In addition to providing a language to analyze mathematical objects and relationships between them, category theory also formalizes the notion of ``structure'' itself~\cite{Corry2004}, thus lending clarity and rigor to situations that may otherwise not be well-understood. Its purview has also extended well beyond mathematics to data analysis where, for instance, it forms the backbone of topological data analysis~\cite{CarlssonTDA} and the dimensionality reduction technique of Uniform Manifold Approximation and Projection (UMAP)~\cite{McInnesUMAP}, and more recently it has been used to provide guiding principles in the context of (classical) machine learning~\cite{gavranovic2024categorical,PeCrKn24}. 

In this article, we provide an overview of various quantum encodings and offer a guided perspective to designing such encodings in a way that respects the structure within a given dataset and learning task. We achieve this through the usage of category theory, which allows one to analyze a variety of mathematical structures simultaneously without the intricate details that are problem-specific. Importantly, the reader is not assumed to have prior familiarity with category theory. Instead, this article serves an invitation to those with a background in QML to adopt a categorical perspective when designing encoding schemes. And although we focus primarily on quantum encodings towards the advancement of QML techniques, the ideas presented here are equally applicable to standard (classical) machine learning. 

The article is outlined as follows. Section~\ref{sec:SPQE} introduces quantum encodings and walks through several examples while highlighting some of the mathematical structures that are preserved. Section~\ref{sec:CatsQE} introduces some basic definitions in category theory and recasts the examples of the previous section in a more formal, categorical context. This will reveal the fact that each example is a special case of a more general concept known as a 
\emph{forgetful functor}. These observations are finally summarized in Section \ref{sec:lifting} in which we reformulate the design setup and design goal of quantum encodings in the language of category theory---the core perspective of the article. 
Section~\ref{sec:OQMF} provides several open questions and avenues for future directions.

\section{Structure-preserving quantum encodings}
\label{sec:SPQE}

Quantum circuits for quantum machine learning tasks can often be decomposed into three essential components: the encoding block, the variational block, and the measurement block, as displayed in Figure \ref{fig:circuit-block}~\cite{schuld2019quantum}. Each component has a significant influence on the output of a quantum algorithm, though the first step plays an particularly important role. Indeed, just as the quality of data influences the performance of machine learning algorithms, it is desirable to preserve  structure inherent in raw data when encoding it onto a quantum computer. And yet, due to the wide variety of quantum algorithms, there are no immediately obvious general guidelines for preserving the latent information within the data in this encoding layer. For this reason, we focus our attention on the encoding step, with special attention given to preserving underlying structure within the data. 

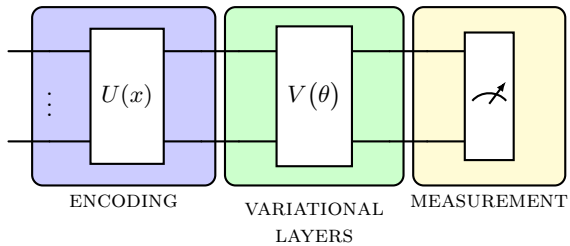
\begin{figure}
\begin{quantikz}[row sep={0.6cm,between origins}]
& \gategroup[3,steps=3,style={solid,rounded
corners,fill=blue!20, inner
xsep=2pt},background,label style={label
position=below,anchor=north,yshift=-0.2cm}]{{\sc
encoding}} & \gate[3]{U (x)} & &  \gategroup[3,steps=3,style={solid,rounded
corners,fill=green!20, inner
xsep=2pt},background,label style={label
position=below,anchor=north,yshift=-0.2cm}]{{\begin{tabular}{c}\sc variational \\ \sc layers\end{tabular}}} & \gate[3]{V \big(\theta \big)} & & 
\gategroup[3,steps=3,style={solid,rounded
corners,fill=yellow!20, inner
xsep=2pt},background,label style={label
position=below,anchor=north,yshift=-0.2cm}]{{\sc measurement}}
& \meter[3]{} & \setwiretype{n}
 \\
\setwiretype{n} & \vdots  &&&& &&& & \setwiretype{n}\\
& &&&& &&& & \setwiretype{n}
\end{quantikz}    
\caption{\label{fig:circuit-block} The three essential components of a (variational) quantum machine-learning task or algorithm: the encoding block that transfers classical data $x$ onto the quantum computer, the variational block whose parameters $\theta$ can be modified so as to optimize some outcome, and the measurement. Our focus here is on the encoding step.}
\end{figure}

As we will see below, the meaning of ``structure'' will depend on the data, the machine learning task, and/or the problem in general. Intuitively, for example, the data could admit geometric or topological structure (such as with data points sampled from a distribution localized near a manifold), algebraic structure (such as symmetries), and/or metric structure (such as when measuring distance based on similarities and dissimilarities between data points). Moreover, many of these structures are not mutually exclusive, which further complicates selecting appropriate quantum encodings. 

Several specialized encoding schemes for transferring classical data onto a quantum computer have been constructed for specific cases. For example, a common choice in the context of original quantum algorithms, such as the Deutsch--Jozsa algorithm~\cite{DeutschJozsa92}, is bit encoding, where a binary string such as $1011$ is assigned to the 4-qubit state $\ket{1011}$. Meanwhile, encoding vectors $v=(v_1,\dots,v_d)$ in $\R^{d}$ into quantum states can be done in many ways. For example solving systems of linear equations~\cite{HHL09}, 
supervised and unsupervised clustering~\cite{LMR13}, 
and computing persistent homology in topological data analysis~\cite{LloydGarneroneZanardi16} 
utilize amplitude encoding when working directly with the raw data~\cite{GroverRudolph02,KayeMosca01}. Other popular choices of encoding schemes in the context of quantum machine learning include angle/rotation encoding, time-evolution encoding, IQP encoding, and more~\cite{SchuldPetruccione21}. Are these encodings the most ideal ones used for their purposes? What guiding principles should be followed so that one can make better-informed decisions about what types of encodings to use? Can structure that is inherent in the problem be utilized in order to guide those choices? This paper aims to answer that this can indeed be done using appropriate category-theoretic tools~\cite{riehl,mac2013categories}. Before getting there, however, we first review some examples of encoding techniques and the types of mathematical structures they preserve, saving the categorical explanations for Section~\ref{sec:catsbg}. 

\subsection{Generic encodings}
  
En route to surveying examples of quantum encodings, we begin by providing the mathematical definition of quantum encoding used throughout this work. The definition is meant to capture the idea that classical real-world data can be mapped into a quantum system either as a data-dependent quantum state or a data-dependent quantum circuit. To make this mathematically precise, we first establish some notation and terminology for operators and states in Hilbert space~\cite{vN18,NiCh11,HallQTM13}.
Given a Hilbert space $\mathcal{H}$, let $\states(\mathcal{H})$ denote the set of states, i.e., trace-class positive operators with trace $1$, and let $\unitary(\mathcal{H})$ denote the set of unitary operators on $\mathcal{H}$. 
Although both $\states(\mathcal{H})$ and $\unitary(\mathcal{H})$ have more structure than that of mere sets (for instance, \ $\states(\mathcal{H})$ can be viewed as a convex space, and $\unitary(\mathcal{H})$ can be equipped with the structure of a Lie group), our aim is to first focus on the structure present in the data and how that structure can be preserved on $\states(\mathcal{H})$ and $\unitary(\mathcal{H})$ under quantum encodings. So, we consider them as mere sets for now. 

We will often let $\mathcal{X}$ denote a set, called a \define{data domain} or \define{feature space}, containing the data of interest, with the notation $X$ typically used for the data set itself (e.g.\ a sample), which is usually taken to be finite. The following definition of a quantum encoding is a minimalistic definition, in that it does not refer to any specific task (e.g.\ classification or regression) or model (e.g.\ variational quantum circuit, implicit, explicit, data re-uploading, etc.) but is general enough to be applicable to most settings. 

\begin{definition}\label{def:encoding}
A \define{quantum state encoding} from a data domain $\mathcal{X}$ into a Hilbert space $\mathcal{H}$ is a function $\rho:\mathcal{X}\to\states(\mathcal{H})$. A \define{quantum unitary encoding} is a function $U:\mathcal{X}\to\unitary(\mathcal{H})$. A \define{quantum encoding} refers to either of these. 
\end{definition}

The two definitions are related in that every quantum unitary encoding $U:\mathcal{X}\to\unitary(\mathcal{H})$, together with a quantum state $\ket{\psi_0}\in\mathcal{H}$, gives rise to a quantum state encoding $\rho:\mathcal{X}\to\states(\mathcal{H})$ via 
\begin{equation}
\rho(x)=U(x)\ket{\psi_0} \bra{\psi_0}U(x)^{\dag}.
\end{equation}

The important point to notice in this general definition is that a quantum encoding is simply a function, which does not necessarily preserve any additional structure. In particular, it is not assumed to be continuous, smooth, distance-preserving, or symmetry-preserving. An important step in designing quantum encodings, therefore, is to \textit{first} identify appropriate structures on the data domain $\mathcal{X}$ and set of states $\states(\mathcal{H})$ that are relevant to the problem at hand. One \textit{then} seeks to construct a quantum encoding that preserves that structure. Demanding that a certain mathematical structure is preserved places a restriction on the set of all functions from the data domain $\mathcal{X}$ to the set of states $\states(\mathcal{H})$, thus potentially making it easier to search for quantum encodings (cf.\ Figure~\ref{fig:reducingspaceofencodings}). 

In the next few sections, we will highlight several familiar illustrations of this, where the relevant structures are symmetry, topologies and smooth structures, distances in the context of topological data analysis, and distances in the context of quantum metric learning. 

\begin{figure}
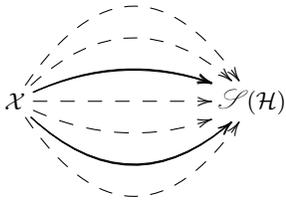

\[
\xy 0;/r.25pc/:
(0,0)*+{\mathcal{X}}="X";
(30,0)*+{\states(\mathcal{H})}="SH";
{\ar@{-->}"X";"SH"};
{\ar@/^1.0pc/@{->}"X";"SH"};
{\ar@/^0.98pc/@{->}"X";"SH"};
{\ar@/^1.02pc/@{->}"X";"SH"};
{\ar@/^2.0pc/@{-->}"X";"SH"};
{\ar@/^3.0pc/@{-->}"X";"SH"};
{\ar@/_1.0pc/@{-->}"X";"SH"};
{\ar@/_2.0pc/@{->}"X";"SH"};
{\ar@/_1.98pc/@{->}"X";"SH"};
{\ar@/_2.02pc/@{->}"X";"SH"};
{\ar@/_3.0pc/@{-->}"X";"SH"};
\endxy
\]
\caption{Among all possible set-theoretic functions describing quantum state encodings (dashed arrows) from a data domain $\mathcal{X}$ to a Hilbert space $\mathcal{H}$, demanding that a structure is preserved isolates a subset of quantum encodings (solid arrows), thus potentially simplifying the search for quantum encodings compatible with a given structure.}
\label{fig:reducingspaceofencodings}
\end{figure}

\subsection{Symmetry}
\label{sec:symmetry}

Geometric quantum machine learning provides one example of structure-preserving quantum encodings that aids in the design of quantum circuits~\cite{larocca2022group,ragone2023representation,meyer2023exploiting,Nguyen_equivariantQNN2024,Glick_covariant_2024}. 
In this paradigm, symmetries are described in terms of group actions on both the data domain and the set of quantum states. 
Quantum encodings that preserves this structure of symmetry are called \emph{equivariant} quantum encodings, and such encodings will be the focus of this section. 
Since there are fewer quantum encodings that are equivariant with respect to these symmetries as compared to all quantum encodings, one is able to isolate a smaller space of encodings, which in theory reduces the resources needed to find appropriate encodings. More precisely, both the data domain $\mathcal{X}$ and quantum state space $\states(\mathcal{H})$ are equipped with an action of a group $G$, and a quantum state encoding $\rho:\mathcal{X}\to\states(\mathcal{H})$ that is compatible with the relevant symmetry is a $G$-equivariant map. Let us now briefly review how this works by first recalling the formal definition of a group's action on a set. 

\begin{definition}
\label{defn:Gset}
A \define{$G$-set} is a set $\mathcal{X}$ together with a group homomorphism $\alpha:G\to\mathrm{Aut}(\mathcal{X})$, where $\mathrm{Aut}(\mathcal{X})$ denotes the group (under function composition) of bijections on $\mathcal{X}$. The map $\alpha$ is referred to as the \define{group action}. Such a $G$-set $\mathcal{X}$ is often written as a pair $(\mathcal{X},\alpha)$.
\end{definition}

In this definition, notice the distinction between the underlying set and the additional structure that defines the action. The notation helps to keep this distinction front and center. Indeed if $\mathcal{X}$ is a $G$-set, then writing it instead as a pair $(\mathcal{X},\alpha)$ helps to emphasize the additional structure that $\mathcal{X}$ carries with it, which is useful notation that will reappear later. Unwinding Definition \ref{defn:Gset} further, ``Aut'' stands for ``automorphisms,'' and the bijection (or automorphism) associated with a group element $g\in G$ is often written as $\alpha_{g}:\mathcal{X}\to\mathcal{X}$. Its action on an element $x\in\mathcal{X}$ is written as $\alpha_{g}(x)$, which is another element of $\mathcal{X}$. The group homomorphism property says that the action is compatible with the group operation, in that $\alpha_{gh}(x)=\alpha_{g}(\alpha_{h}(x))$ for all $g,h\in G$ and $x\in\mathcal{X}$ as well as $\alpha_{1_{G}}=\id_{\mathcal{X}}$, where $1_{G}$ is the identity in the group $G$. 

As an example, if $\mathcal{H}$ is a Hilbert space, and if $V: G\to\unitary(\mathcal{H})$ is a unitary representation of $G$ on $\mathcal{H}$, then $\mathcal{H}$ becomes a $G$-set because the set of unitary operators on $\mathcal{H}$ can be viewed as a subset of $\mathrm{Aut}(\mathcal{H})$, since unitary operators can be viewed as invertible linear transformations on $\mathcal{H}$. Moreover, this action of $G$ on $\mathcal{H}$ induces an action of $G$ on the space of states $\states(\mathcal{H})$ by sending each group element $g\in G$ to the operator $\Ad_{V_g}\in\mathrm{Aut}(\states(\mathcal{H}))$ defined by 
\begin{equation}
\Ad_{V_g}(\sigma):=V_g\sigma V_g^{\dag}
\end{equation}
for all $\sigma\in\states(\mathcal{H})$. Here ``$\mathrm{Ad}$'' stands for the \emph{adjoint} action. Thus, $(\states(\mathcal{H}),\Ad_{V})$ is a $G$-set as well. 

Therefore, if we have a data domain $\mathcal{X}$ with a symmetry, and we want to encode that data onto a quantum system with Hilbert space $\mathcal{H}$, then we might want to encode that data in a way that preserves the symmetry. This leads to our first example of a structure-preserving quantum encoding.  

\begin{definition}
\label{defn:Gequivariantquantumencoding}
Given a $G$-set $(\mathcal{X},\alpha)$, a \define{$G$-equivariant quantum state encoding} of $(\mathcal{X},\alpha)$ consists of a quantum state encoding $\rho:\mathcal{X}\to\states(\mathcal{H})$ together with a unitary representation $V: G\to\unitary(\mathcal{H})$ of $G$ on the Hilbert space $\mathcal{H}$ such that $\rho$ is \define{$G$-equivariant}, that is,
    \begin{equation}
    \label{eqn:Gequivariantdefn}
    \rho\big(\alpha_{g}(x)\big)=V_{g}\rho(x)V_{g}^{\dagger}
    \end{equation}
    for all $x\in \mathcal{X}$ and $g\in G$.
\end{definition}

This definition of a $G$-equivariant quantum state encoding is a special case of the more general notion of a $G$-equivariant map between two $G$-sets, defined as follows. 

\begin{definition}
\label{defn:Gequivmap}
Let $(\mathcal{X},\alpha)$ and $(\mathcal{Y},\beta)$ be two $G$-sets. A \define{$G$-equivariant map} from $(\mathcal{X},\alpha)$ to $(\mathcal{Y},\beta)$ is a function $f:\mathcal{X}\to\mathcal{Y}$ such that 
\begin{equation}
\label{eq:Gsetmap}
f\big(\alpha_{g}(x)\big)=\beta_{g}\big(f(x)\big)
\end{equation}
for all $x\in\mathcal{X}$ and $g\in G$. 
\end{definition}

As we will see in Section \ref{sec:CatsQE}, Definition \ref{defn:Gequivariantquantumencoding} amounts to saying that a $G$-equivariant quantum state encoding is a \emph{morphism} in a particular \emph{category}. In  the meantime, observe that Equation~\eqref{eq:Gsetmap}, and in particular Equation~\eqref{eqn:Gequivariantdefn}, imposes a restriction on the set of all functions from one $G$-set to another since not every function will necessarily satisfy the constraint. Indeed, such constraints have been used in QML to construct equivariant quantum encodings by combining equivariant feature maps and equivariant gatesets in Ref.~\cite{meyer2023exploiting}. This is illustrated in the following toy example. 

\begin{example}
\label{ex:Meyermodified}
Following the example in Ref.~\cite[Section II.A]{meyer2023exploiting}, consider a supervised learning task that distinguishes between two classes of points within a data domain $\mathcal{X}=\mathbb{R}^2$ as shown in Figure~\ref{fig:meyersymmetrymodified}. 
The data domain $\mathcal{X}$ acquires a symmetry determined by the relations
\begin{equation}
\label{eqn:symmetryKlein4groupclassifier1}
y(x_1,x_2)=y(x_2,x_1)=y(-x_1,-x_2),
\end{equation}
where $y:\mathcal{X}\to\R$ is a function that determines the class of any input data point via $y(x)>0$ implies $x$ is of class $+1$, while $y(x)<0$ implies $x$ is of class $-1$. The symmetry depicted in~\eqref{eqn:symmetryKlein4groupclassifier1} is modeled by the Klein Four group 
\begin{equation}
G=\mathbb{Z}_2\times\mathbb{Z}_2=\{(0,0),(1,0),(0,1),(1,1)\}
\end{equation}
under addition modulo $2$. 
Suppose this group has the action $\alpha: G\to \mathrm{Aut}(\R^2)$ on $\mathcal{X}=\R^{2}$ given by the linear transformations whose associated matrices are
\begin{equation}
\begin{aligned}[c]
\alpha_{(0,0)} &= \begin{bmatrix}1&0\\0&1\end{bmatrix}\\[5pt]
\alpha_{(1,0)} &= \begin{bmatrix}0&1\\1&0\end{bmatrix}
\end{aligned}
\qquad 
\begin{aligned}[c]
\alpha_{(0,1)} &= \begin{bmatrix}-1&0\\0&-1\end{bmatrix}\\[5pt]
\alpha_{(1,1)} &= \begin{bmatrix}0&-1\\-1&0\end{bmatrix}.
\end{aligned}
\end{equation}
Now set $\mathcal{H}=\C^{2}\otimes\C^{2}$ and set $V: G\to\unitary(\mathcal{H})$ to be the representation specified by
\begin{equation}
\begin{aligned}[c]
V_{(0,0)}&=\mathds{1}_{2}\otimes \mathds{1}_{2} \\[5pt]
V_{(1,0)}&= \SWAP 
\end{aligned}
\qquad 
\begin{aligned}[c]
V_{(0,1)}&= X\otimes X \\[5pt]
V_{(1,1)}&= \SWAP(X\otimes X),
\end{aligned}
\end{equation}
where $X$ is the Pauli matrix $X=\left[\begin{smallmatrix}0&1\\1&0\end{smallmatrix}\right]$, $\mathds{1}_{2}$ is the $2\times2$ identity matrix, and $\SWAP$ is the unitary operator characterized by $\SWAP \ket{\psi}\otimes\ket{\phi}=\ket{\phi}\otimes\ket{\psi}$ for all $\ket{\psi},\ket{\phi}\in\C^{2}$. 

In this example, a $G$-equivariant quantum encoding would be a map $\rho:\mathcal{X}\to\states(\mathcal{H})$ such that $V_{g}\rho(x)V_{g}^{\dag}=\rho(\alpha_{g}(x))$ for all $x\in\mathcal{X}$ and $g\in G$. One example of such an encoding can be obtained from the unitary encoding $U:\mathcal{X}\to\unitary(\mathcal{H})$, which sends a point $x=(x_1,x_2)\in \mathcal{X}$ to 
\begin{equation}
\label{eqn:embeddingUsymmetryexample}
U(x_1,x_2)= R_Z(x_1)\otimes R_Z(x_2),
\end{equation}
where the rotation gate is given by
\begin{equation}
R_Z(\theta)=
e^{-i\theta Z/2}=
\begin{bmatrix}
    e^{-i\theta/2} & 0 \\ 0 & e^{i\theta/2}
\end{bmatrix},
\end{equation}
where $Z$ is the Pauli matrix $Z=\left[\begin{smallmatrix}1&0\\0&-1\end{smallmatrix}\right]$. 
From this unitary encoding, one can obtain a state encoding by acting on a fiducial state $\ket{\psi_{0}}\in\mathcal{H}$ of the form 
\begin{equation}
\label{eqn:psi0symmetrytoyexample}
\ket{\psi_{0}}=\sqrt{p}\ket{+,+}-\sqrt{1-p}\ket{-,-},
\end{equation}
where $0<p<1$.  
(The value of $p$ chosen to produce Figure~\ref{fig:meyersymmetrymodified} can be found in Appendix~\ref{app:GQML}.) Here, 
\begin{equation}
\ket{+}=\frac{1}{\sqrt{2}}\big(\ket{0}+\ket{1}\big)
\;\;\text{ and }\;\;
\ket{-}=\frac{1}{\sqrt{2}}\big(\ket{0}-\ket{1}\big)
\end{equation}
are the spin-up and spin-down eigenvectors of $X$, respectively, and $\ket{\pm,\pm}:=\ket{\pm}\otimes\ket{\pm}$.
Thus, the resulting state encoding sends $x=(x_1,x_2)$ to 
\begin{equation}
\rho(x)=U(x_1,x_2)\ket{\psi_0}\bra{\psi_0}U(x_1,x_2)^{\dag}.
\end{equation}

    \begin{figure}
    \includegraphics{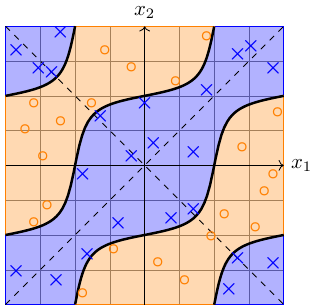}
    \caption{This is a variation of Figure 2 in Ref.~\cite{meyer2023exploiting} based on a binary classification task with discrete symmetries given by $\alpha_{(1,0)}$ a reflection along the line $x_{2}=x_{1}$, $\alpha_{(0,1)}$ an inversion, and $\alpha_{(1,1)}$ a reflection along the line $x_{2}=-x_{1}$. The structure is periodic so that the symmetry is preserved. More details about how to explicitly construct the decision boundaries for this classifier, as well as the associated observable, are provided in Appendix~\ref{app:GQML}.}
    \label{fig:meyersymmetrymodified}
    \end{figure}

One can explicitly check that this defines an \emph{equivariant} quantum encoding. More generally, one might ask how such an equivariant encoding can be obtained from the equivariance condition. To see how demanding equivariance reduces the space of possible quantum encodings, it helps to make some assumptions about the form of the encoding. Suppose that the unitary quantum encoding $U:\mathcal{X}\to\unitary(\mathcal{H})$ is of the form 
\begin{equation}
\label{eqn:unitary4by4example}
U(x)=e^{-i \mathcal{L}(x)},
\end{equation}
where $\mathcal{L}:\R^{2}\to\matr_{4}$ is a linear transformation sending each point $x\in\mathcal{X}$ to a Hermitian operator $\mathcal{L}(x)\in\matr_{4}$, where $\matr_{4}$ is the algebra of complex $4\times$ matrices. (Encodings of this form are described in more detail in Appendix~\ref{app:onepgroupqencode}.) Moreover, it suffices to assume this Hermitian operator is traceless so that $-i\mathcal{L}(x)\in\mathfrak{su}(4)$, the Lie algebra of $4\times4$ special unitary matrices. Such a linear transformation $\mathcal{L}$ is uniquely determined by the image of the two unit vectors 
\begin{equation}
L_{1}:=\mathcal{L}(e_1)
\quad\text{ and }\quad
L_{2}:=\mathcal{L}(e_2),
\end{equation}
where $e_1=(1,0)$ and $e_2=(0,1)$. Therefore, the set of all such linear transformations is isomorphic to the span of all ordered pairs of traceless Hermitian $4\times 4$ matrices, which is a real vector space of dimension $2(4^2-1)=30$. 

The equivariance condition singles out a subspace of this vector space as follows. First notice that the equivariance condition 
\begin{equation}
V_{g}e^{-i\mathcal{L}(x)}V_{g}^{\dag}=e^{-i\mathcal{L}(\alpha_g(x))}
\end{equation}
at the level of the unitary encoding becomes 
\begin{equation}
V_g \mathcal{L}(x) V_{g}^{\dag}=\mathcal{L}(\alpha_g(x))
\end{equation}
at the level of the generating elements $\mathcal{L}(x)$ because each $V_g$ is unitary~\cite{hall2015lie}. Let us now see what this says for some of the specific group elements. For $g=(1,0)$, equivariance implies $\SWAP L_1 \SWAP =L_2$ (and also $\SWAP L_2 \SWAP =L_1$, but this is equivalent to the first equation). For $g=(0,1)$, equivariance implies $(X\otimes X)L_1(X\otimes X)=-L_1$ and $(X\otimes X)L_2(X\otimes X)=-L_2$. However, the second of these follows from the previous constratins. Hence, these give us two \emph{operator} constraints
\begin{align}
\SWAP L_1 \SWAP&=L_{2} \nonumber\\
\{X\otimes X,L_{1}\}&=0 \label{eq:GMLexconstraint} 
\end{align}
on the space of all quantum encodings parameterized by ordered pairs of Hermitian operators, thereby reducing the space of all possible quantum encodings to structure-preserving quantum encodings. 
In fact, there is a reduction from the 30-dimensional real vector space of all pairs of Hermitian matrices $(L_1,L_2)$ to an 8-dimensional subspace of pairs $(L_1,\SWAP L_1 \SWAP)$ due to the constraint~\eqref{eq:GMLexconstraint}. A basis for this subspace in terms of $L_{1}$ consists of 
\begin{equation}
\begin{aligned}[c]
&Z\otimes\mathds{1}_{2},\;\;  & \mathds{1}_{2}\otimes Z,\;\; \\[5pt]
&Y\otimes\mathds{1}_{2},
&\mathds{1}_{2}\otimes Y,\;\;
\end{aligned}
\qquad 
\begin{aligned}[c]
&X\otimes Z,\;\;
&Z\otimes X,\;\; \\[5pt]
&X\otimes Y,\;\; 
&Y\otimes X,\;\;
\end{aligned}
\end{equation}
which is derived in Appendix~\ref{app:GQML}.
The example given in~\eqref{eqn:embeddingUsymmetryexample} from Ref.~\cite{meyer2023exploiting} corresponds to the choice
\begin{equation}
L_{1}=\frac{Z}{2}\otimes\mathds{1}_{2}
\quad\text{ and }\quad
L_{2}=\mathds{1}_{2}\otimes\frac{Z}{2}
\end{equation}
for the generators of the quantum encoding. 

More details associated with this example, including how the decision boundaries in Figure~\ref{fig:meyersymmetrymodified} are constructed, are provided in Appendix~\ref{app:GQML}. 
\end{example}

There are two subtle points to be aware of when designing equivariant encodings. One is that perfectly equivariant encodings (and subsequent equivariant layers) may overfit to perfectly symmetric data. The second point is that such equivariance may also lead to efficient classical replacements of quantum algorithms, i.e., \emph{dequantization}~\cite{Tang2022dequantizing,Anschuetz2023}. For the first point, due to errors on quantum computers and other factors, the overfitting to perfectly symmetric data might inhibit generalization, and might therefore \textit{not} be ideal for all QML problems~\cite{tuysuz2024symmetry}. Indeed, it has been observed that having symmetry-breaking layers in data re-uploading models may provide more accurate predictions in classification tasks~\cite{larocca2022group,LeIsabel2023,langer2024probing}. (We only briefly mention data re-uploading models in Appendix~\ref{app:metriclearningsms} but otherwise do not make explicit use of them.) However, it is important even in these examples to first isolate the class of equivariant encodings that one can then perturb to get symmetry-breaking gates in a way that is informed by the types of errors that can occur on the quantum computing device. Therefore, isolating the equivariant encodings and gates is an important first step in designing such quantum machine learning models. A discussion about the second point on dequantization is deferred to Section~\ref{sec:OQMF}. 

Having briefly presented geometric quantum machine learning as the archetypal example of how structure (namely, symmetry) in a dataset could be used to inform the design of quantum encodings, we now move on to discuss structure preservation in the context of continuity and smoothness. 

\subsection{Topologies and Smooth Structures}
\label{sec:top_smooth}

Continuity and smoothness (eg.\ differentiability) are concepts that can be defined using the tools of topology and manifold theory~\cite{Mu00,BaMu94,LeeISM13,BBTtopology20}, which isolate some of the more flexible structures of Euclidean space dictating how neighborhoods connect together~\cite{Ru76}. 
The underlying topologies and smooth structures on the data domains and quantum state spaces are two examples of mathematical structures that are so natural they might not always be explicitly stated. Quantum encodings respecting these structures are functions that are continuous and smooth, respectively. 
In more detail, the data domain $\mathcal{X}$ and the state space $\states(\mathcal{H})$ should each be equipped with a topology and a smooth structure so that the quantum state encoding $\rho:\mathcal{X}\to\states(\mathcal{H})$ can be chosen to be continuous or smooth (i.e., infinitely differentiable) if we want the encoded quantum states to vary continuously or smoothly based on small variations in the classical data. 

We will not need the technicalities of these definitions in what follows, but we do wish to reinforce the notion of \textit{adding mathematical structure to a set}, which is key concept in this work. To that end, let us briefly recall the definition of a topology and a smooth structure. A \define{topology} on a set $\mathcal{X}$ is a set $\tau_\mathcal{X}$ whose elements are certain subsets of $\mathcal{X}$, called \define{open} subsets, satisfying a list of axioms~\cite{Mu00,BBTtopology20}. When a set $\mathcal{X}$ is equipped with a topology $\tau_\mathcal{X}$, the pair $(\mathcal{X},\tau_\mathcal{X})$ is called a \define{topological space}. Writing it as a pair  emphasizes the additional structure that $\mathcal{X}$ carries, although $\tau_\mathcal{X}$ is often omitted in the literature for brevity. Meanwhile, a smooth structure can be defined once a set $\mathcal{X}$ is already equipped with a topology. Moreover, this topology should satisfy an additional list of axioms so that $(\mathcal{X},\tau_\mathcal{X})$ defines a \define{topological manifold} (namely, it should be Hausdorff, second-countable, and locally Euclidean of some fixed dimension)~\cite{LeeISM13}. Then, a \define{smooth structure} on a topological manifold $(\mathcal{X},\tau_\mathcal{X})$ is a maximal smooth atlas $\mathcal{A}$ on $\mathcal{X}$, and the triple $(\mathcal{X},\tau_\mathcal{X},\mathcal{A})$ is called a \define{smooth manifold}. (See Ref.~\cite{LeeISM13} for more details.) Given such structures, only then can one define \define{continuous} and \define{smooth} maps between topological spaces and smooth manifolds as functions that preserve the topology and smooth structures~\cite{Mu00,BBTtopology20,LeeISM13}. In the language of category theory, these maps are said to be \emph{morphisms} in particular \emph{categories}, a framing we will revisit in Section \ref{sec:CatsQE}. 

Technicalities aside, continuity of a quantum encoding is used to ensure that classical data is encoded in such a way that the encoded data do not vary too wildly from the raw classical data, while smoothness is used to ensure even more rigidity and differentiability properties, such as when calculating generators using derivatives~\cite{SchuldPetruccione21}. (See also Appendix~\ref{app:onepgroupqencode}.) Most quantum encodings are of this form, and in this section we will list some of the more commonly-used examples in the literature. In each of these examples, the data domain $\mathcal{X}$ is a Euclidean space of the form $\R^{d}$, which is equipped with the standard topology and smooth structure of Euclidean space~\cite{LeeISM13}. Meanwhile, the set of pure states on a Hilbert space $\mathcal{H}=\C^{d}$ is the complex projective space $\CP^{d-1}$, which also has a standard topology and smooth structure~\cite{LeeISM13}, %
\footnote{Defining smoothness on the full state space $\states(\mathcal{H})$ is a bit subtle because $\states(\mathcal{H})$ is not technically a smooth manifold, though it is a convex space~\cite{BeZy17}. Despite the lack of smoothness in the usual sense, for most of the examples we consider in this paper, the quantum state encodings will factor through the subspace of pure states, which is a smooth manifold (it is a complex projective space). For this reason, we will not concern ourselves too much with going into the more technical details regarding the smooth structure on $\states(\mathcal{H})$.}. 

\begin{example}[Angle encoding]
\emph{Angle encoding}, also called \emph{rotation encoding}, is a technique explored in Refs.~\cite{schuld21,schuld2019quantum, skolik2021layerwise} that utilizes the structure of one-parameter groups and interprets the feature components as parameters that can be used for Hamiltonian evolution. 
In the case that $\mathcal{X}=\R^{d}$ is a Euclidean feature space, each component of $x=(x_1,\ldots,x_d)\in\mathcal{X}$ is mapped to its own qubit via the map
\begin{equation}
x \mapsto \ket{\Phi(x)} := \left( \bigotimes_{k=1}^{d} \exp\left( - \frac{i x_k}{2} X_k  \right) \right) \ket{0\cdots0},
\end{equation}
where $X_{k}$ is the Pauli $X$ gate acting on the $k^{\text{th}}$ qubit. 
The associated quantum state encoding $\rho:\R^{d}\to\states(\C^{2^d})$ is then given by 
\begin{equation}
\rho(x):=\ket{\Phi(x)}\bra{\Phi(x)}. 
\end{equation}
Note that the Pauli $X$ gate may be replaced with a Pauli $Y$ or Pauli $Z$ gate, as in Equation \eqref{eqn:embeddingUsymmetryexample} in the previous section, for example. In fact, one could choose any generalized Pauli operator $\vec{n}\cdot\vec{\sigma}$, where $\vec{n}$ is some unit vector in $\R^3$ and $\vec{\sigma}=(X,Y,Z)$ is the vector of the Pauli gates. All of these quantum state encodings are continuous and smooth. Note that these quantum encodings map a feature space of $d$ dimensions into $d$ qubits, which is a linear scaling (as opposed to the logarithmic scaling of amplitude encoding, which is discussed in Example~\ref{ex:amplitudeencoding}). Moreover, angle encoding has been successfully used in some binary classification tasks~\cite{lloyd20,KiBe22,de2024empirical}.  Note, however, that angle encodings do not make any use of entanglement with respect to the assumed tensor factorization of the qubits because there are no entangling gates in the definition of the encoding. For this reason, it is said to have low expressibility~\cite{sim2019expressibility}, and one can include entangling gates and data re-uploading within the quantum circuit to increase expressibility and expressivity~\cite{Jerbi2023}. 
\end{example}

\begin{example}[Amplitude encoding]
\label{ex:amplitudeencoding}
\emph{Amplitude encoding} is a technique that models entries in the data point array as frequencies (in the statistical sense) of the expected string output~\cite{SchuldPetruccione21}. By definition, it takes a vector $x\in\R^{2^d}\setminus\{0\}$ to a vector $\chi(x)$ in $\C^{2^d}$, the Hilbert space of $d$ qubits, and then normalizes it to obtain a genuine quantum state $\ket{\Phi(x)}$. To express this as a succinct formula, if $x=(x_0,x_1,\dots,x_{2^d-1})\in\R^{2^d}$, then rewrite the indices in binary so that they correspond to a sequence $n$ of $d$ $0$'s and $1$'s. For example, if $d=2$, then $n$ has four possibilities: $n\in\{00,01,10,11\}$. Then amplitude encoding is defined by the sequence of maps
\begin{equation}
\label{eq:ampencOG}
x\mapsto \chi(x)=\sum_{n=0}^{2^d-1}x_{n}\ket{n}\mapsto \ket{\Phi(x)}=\frac{\chi(x)}{\lVert\chi(x)\rVert},
\end{equation}
where $\lVert\chi(x)\rVert$ denotes the Euclidean norm of $\chi(x)$. Since the data is assumed to be real, one could equivalently normalize the data first onto the sphere $S^{2^{d}-1}$ inside $\R^{2^d}$ and then map over to $\C^{2^d}$ to get the same resulting quantum state.
Although the above description of amplitude encoding provides one mathematical description of the resulting function, there are several proposals for constructing quantum circuits that achieve amplitude encoding, not all of which use the same number of qubits~\cite{GroverRudolph02,KayeMosca01,araujo2021divide}. We will not discuss these at present, but instead analyze some of the structural properties of the amplitude encoding from~\eqref{eq:ampencOG}. 
To start, amplitude encoding is both continuous and smooth. Moreover, one of the benefits of amplitude encoding is that a feature space of the form $\R^{2^d}$ is encoded into a Hilbert space of $\log_{2}(2^d)=d$ qubits. However, in the process of normalization to obtain a genuine quantum state, data points that are far away from each other in the data domain might have their distances so drastically reduced that they become difficult to separate in the encoding (cf.\ Figure~\ref{fig:ampenccompress}). In fact, the mapping is in general not one-to-one, which means that two distinct points may be mapped to the same element, thus causing their distance to vanish and therefore become indistinguishable. 

\begin{figure}
\begin{tikzpicture}
\def\lef{-1}
\def\righ{4}
\draw (\lef,0) circle (1.0);
    \node at ({0.3*cos(deg(20))+\lef},{0.3*sin(deg(20))}) {$\blacklozenge$};
    \node at ({0.1*cos(deg(60))+\lef},{0.1*sin(deg(60))}) {\tiny $\blacksquare$};
    \node at ({0.8*cos(deg(120))+\lef},{0.8*sin(deg(120))}) {$\bullet$};
    \node at ({0.6*cos(deg(125))+\lef},{0.6*sin(deg(125))}) {$\blacktriangle$};
    \node at ({0.7*cos(deg(40))+\lef},{0.7*sin(deg(40))}) {$\bullet$};
    \node at ({0.9*cos(deg(30))+\lef},{0.9*sin(deg(30))}) {$\bullet$};
    \node at ({1.1*cos(deg(50))+\lef},{1.1*sin(deg(50))}) {$\bullet$};
    \node at ({1.4*cos(deg(70))+\lef},{1.4*sin(deg(70))}) {$\bullet$};
    \node at ({1.5*cos(deg(90))+\lef},{1.5*sin(deg(90))}) {$\bullet$};
    \node at ({1.8*cos(deg(100))+\lef},{1.8*sin(deg(100))}) {$\blacktriangledown$};
    \node at ({1.9*cos(deg(110))+\lef},{1.9*sin(deg(110))}) {$\bullet$};
\draw[->] ({\lef+2},0) -- node[above]{\tiny normalize} ({\righ-2},0);
    \node at ({cos(deg(20))+\righ},{sin(deg(20))}) {$\blacklozenge$};
    \node at ({cos(deg(60))+\righ},{sin(deg(60))}) {\tiny $\blacksquare$};
    \node at ({cos(deg(120))+\righ},{sin(deg(120))}) {$\bullet$};
    \node at ({cos(deg(125))+\righ},{sin(deg(125))}) {$\blacktriangle$};
    \node at ({cos(deg(40))+\righ},{sin(deg(40))}) {$\bullet$};
    \node at ({cos(deg(30))+\righ},{sin(deg(30))}) {$\bullet$};
    \node at ({cos(deg(50))+\righ},{sin(deg(50))}) {$\bullet$};
    \node at ({cos(deg(70))+\righ},{sin(deg(70))}) {$\bullet$};
    \node at ({cos(deg(90))+\righ},{sin(deg(90))}) {$\bullet$};
    \node at ({cos(deg(100))+\righ},{sin(deg(100))}) {$\blacktriangledown$};
    \node at ({cos(deg(110))+\righ},{sin(deg(110))}) {$\bullet$};
\draw (\righ,0) circle (1.0);
\end{tikzpicture}
\caption{The standard form of amplitude encoding is smooth but not dimension-preserving. This is because it requires a step that pushes the set of data points in $\R^{2^{d}}$ onto the unit sphere $S^{2^{d}-1}$ thereby reducing the separation between the data points and hence increasing the difficulty in distinguishing between them. For example, although the square {\tiny $\blacksquare$} is close to the diamond $\blacklozenge$ in the ambient data domain, they are far apart after amplitude encoding. Meanwhile, although $\blacktriangledown$ and $\blacktriangle$ are far apart in the ambient data domain, they are close together after amplitude encoding.}
\label{fig:ampenccompress}
\end{figure}
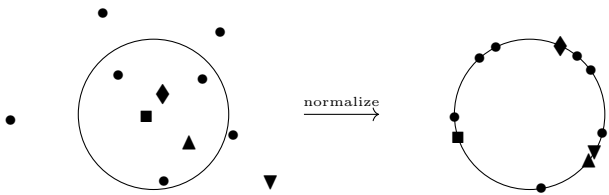

There are at least two slight modifications to the amplitude encoding scheme mentioned above that partially address the lack of injectivity. First, one could add an additional entry for the original data and project onto the unit sphere $S^{2^d}$ (instead of $S^{2^d-1}$) via
\begin{equation}
(x_0,x_1,\dots,x_{N-1})\mapsto \frac{(x_0,x_1,\dots,x_{N-1},1)}{\sqrt{1+\lVert x\rVert^2}}
\end{equation}
so that the dimension does not reduce upon encoding~\cite{SchuldPetruccione21}. Although the map is now a diffeomorphism onto the open set for the northern hemisphere of $S^{2^d}$, this does not fully resolve the problem of the difficulty of separation of data mentioned in Figure~\ref{fig:ampenccompress}. This is because data points that are within the unit sphere $S^{2^d-1}$ are mapped into the top part of the northern hemisphere of $S^{2^d}$ within a polar angle of $\frac{\pi}{4}$, while the points outside the unit sphere $S^{2^d-1}$ are all mapped into the northern hemisphere of $S^{2^d}$ with a polar angle between $\frac{\pi}{4}$ and $\frac{\pi}{2}$, thus drastically reducing the distances between points that were originally very far away from each other (cf.\ Figure~\ref{fig:ampenccompressplusone}).

\begin{figure}
\begin{tikzpicture}
\draw ({-6-3/3},0) -- ({-2-1/3},0);
\draw (0,0) circle (1.0);
\draw[->] (-2.2,0.4) to [out=30,in=150] (-1.3,0.4);
\foreach \x in {-6,-5,-4,-3,-2,-1,0,1,2,3,4,5,6} {
	\node at ({-4+(\x-2)/3},0) {$\bullet$};
	\node at ({\x/(pow(1+\x*\x,1/2))},{1/(pow(1+\x*\x,1/2))}) {$\bullet$};
	}
\foreach \x in {-6,0,6} {
    \node at ({-4+(\x-2)/3},-0.35) {$\x$};
    }
\end{tikzpicture}
\caption{Although data $x_0$ in the interval $[-1,1]$ gets mapped into the northern hemisphere of $S^{1}$ within polar angle $\frac{\pi}{4}$, all other data gets mapped to the region with polar angle between $\frac{\pi}{4}$ and $\frac{\pi}{2}$. This is illustrated by showing how the set of negative integers becomes a sequence with a limit point at $(-1,0)$ and similarly the image of the positive integers has a limit point at $(1,0)$. This is shown here for $[-6,6]\cap\Z$.}
\label{fig:ampenccompressplusone}
\end{figure}
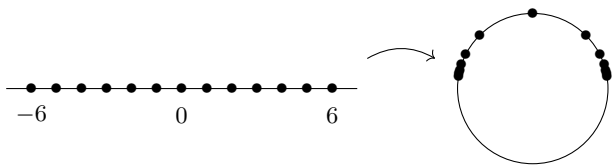

Another version of amplitude encoding instead first transforms the data into a probability vector by the mapping
\begin{equation}
(x_0,\dots,x_{N-1})\mapsto \frac{\left(e^{x_0},\dots,e^{x_{N-1}},1\right)}{1+\sum_{j=0}^{N-1}e^{x_j}},
\end{equation}
which has image on the probability simplex~\cite{vlasic2023qtda}. The benefit of this over the previous encodings is the direct translation to probabilities, though a similar problem of data separation (due to a lack of distance preservation) remains (cf.\ Figure~\ref{fig:ampenccompressplusoneAndrew}). Alternative versions of amplitude encodings are given in Ref.~\cite{Arnottetal2024}, though those also have similar advantages and disadvantages to the examples just given. These problems involving a lack of distance preservation can be circumvented by equipping the target spaces with a modified notion of distance, a topic which we discuss in the next two sections. However, it is not generally possible to transfer a distance function \emph{from} a data domain \emph{to} a quantum state space (since encodings are not generally bijective). Instead, it is only possible to go in the opposite direction (cf.\ Section~\ref{ssec:dist_metriclearning}). Let us therefore first move on to discuss distances in the context of topological data analysis.

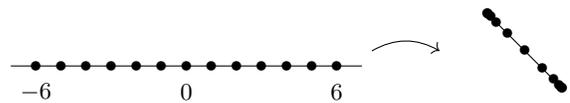
\begin{figure}
\begin{tikzpicture}
\draw ({-6-1/3},0.3) -- ({-2+1/3},0.3);
\draw[->] ({-2.2+2/3},0.5) to [out=30,in=150] ({-1.3+2/3},0.5);
\draw (0,1) -- (1,0);
\foreach \x in {-6,-5,-4,-3,-2,-1,0,1,2,3,4,5,6} {
	\node at ({-4+(\x)/3},0.3) {$\bullet$};
	\node at ({exp(\x)/(1+exp(\x))},{1/(1+exp(\x))}) {$\bullet$};
	}
\foreach \x in {-6,0,6} {
    \node at ({-4+(\x)/3},-0.05) {$\x$};
    }
\end{tikzpicture}
\caption{Data $x_0$ in $\R$ gets mapped into the probability $1$-simplex. Although values near $0$ stay separated, data far away from $0$ gets clustered near the endpoints. The image of the integers is shown for $[-6,6]\cap\Z$. 
}
\label{fig:ampenccompressplusoneAndrew}
\end{figure}

\end{example}

\subsection{Distance for topological data analysis}
\label{ssec:TDA}

Topological data analysis (TDA) is a technique to identify geometric features of a data-generating source that are preserved under continuous perturbations~\cite{CarlssonTDA,carlsson2021topological,ghrist2008barcodes,GlassVidaurre24}. 
A quantum encoding transfers the data to a quantum system, which potentially disrupts these topological features. This section will provide guidance on how to minimally disrupt these features by quantifying the distances between data points before and after the encoding. 

In more detail, TDA utilizes data embedded in a metric space in order to obtain topological features of a postulated underlying space from which the data are generated. Although a finite set of data points does not inherently acquire a nontrivial topology, the distances between points can be used to construct a simplicial complex~\cite{Munkres84}, called the Vietoris--Rips complex, which is a combinatorial object that contains topological information about this underlying space. We will not need to go through all the components of TDA to get the basic flavor of quantum encodings that can be used for TDA. The important concepts to focus on for our purposes are the structures that are involved in TDA, namely distance functions, i.e. metrics, whose definition we recall below~\cite{BuBuIv01,Bryant85}. In this definition, notice how our theme  of \textit{sets equipped with additional structure} makes another appearance.

\begin{definition}
\label{defn:metrispace}
A \define{metric space} is a pair $(\mathcal{X},d_{\mathcal{X}})$, where $\mathcal{X}$ is a set and $d_{\mathcal{X}}:\mathcal{X}\times \mathcal{X}\to\R$ is a function satisfying the following conditions:
\begin{enumerate}[i.]
\item $d_{\mathcal{X}}(x_1,x_2)=0$ if and only if $x_1=x_2$ 
\item $d_{\mathcal{X}}(x_1,x_2)=d_{\mathcal{X}}(x_2,x_1)$ for all $x_1,x_2\in \mathcal{X}$
\item $d_{\mathcal{X}}(x_{1},x_{2})\le d_{\mathcal{X}}(x_{1},x_{3})+d_{\mathcal{X}}(x_2,x_3)$ for all $x_1,x_2,x_3\in \mathcal{X}$ (the \emph{triangle inequality}). 
\end{enumerate}
Such a function $d_{\mathcal{X}}$ is called a \define{distance function}, or \define{metric}, on $\mathcal{X}$. 
\end{definition}

The general pipeline for TDA is as follows (see Figure~\ref{fig:TDApipeline} for a visual flow-chart). 
\begin{enumerate}
\item
If $X\subseteq\mathcal{X}$ is a finite subset of $\mathcal{X}$, which is to be interpreted as a finite sample of data points, and $\mathcal{X}$ has a metric $d_{\mathcal{X}}$ on it, then $X$ acquires a metric $d_{X}$ by the formula $d_{X}(x_1,x_2)=d_{\mathcal{X}}(x_1,x_2)$ for all $x_1,x_2\in X$. 
\item
From $(X,d_X)$, one can construct a combinatorial object called a filtered simplicial complex $K_{\bullet}(X,d_{X})$, with one example being the Vietoris--Rips complex~\cite{Hausmann96,CarlssonTDA}.
\item
From the filtered simplicial complex, one can construct algebraic objects, which are the associated persistence vector spaces and homologies $H_{\bullet}(X,d_X)$~\cite{CarlssonTDA,ZoCa05,Carlsson2014}. 
\item 
From the persistence vector spaces, one can construct numerical quantities, the persistence diagrams $\mathrm{dgm}_{\bullet}(X,d_X)$~\cite{CSEHstability2007}.
\end{enumerate}

\begin{figure}
\xy 0;/r.25pc/:
(0,0)*+{(X,d_X)}="X";
(0,-9)*+{\begin{tabular}{c}finite\\metric\\space\end{tabular}};
(20,0)*+{K_{\bullet}(X,d_X)}="K";
(20,-9)*+{\begin{tabular}{c}filtered\\simplicial\\complex\end{tabular}};
(42,0)*+{H_{\bullet}(X,d_X)}="H";
(42,-9)*+{\begin{tabular}{c}persistence\\vector\\space\end{tabular}};
(66,0)*+{\mathrm{dgm}_{\bullet}(X,d_X)}="dgm";
(66,-9)*+{\begin{tabular}{c}persistence\\diagram\end{tabular}};
{\ar@{|->}"X";"K"};
{\ar@{|->}"K";"H"};
{\ar@{|->}"H";"dgm"};
\endxy
\caption{The persistent homology pipeline uses data and their distances $(X,d_X)$ to first construct a combinatorial object $K_{\bullet}(X,d_X)$, which is then used to construct an algebraic object $H_{\bullet}(X,d_X)$, which is finally used to construct numerical quantities $\mathrm{dgm}_{\bullet}(X,d_X)$ (see text for more details)~\cite{MemoliSinghal19}.}
\label{fig:TDApipeline}
\end{figure}

The original standard \emph{quantum} TDA protocol~\cite{LloydGarneroneZanardi16} jumps in once the filtered simplicial complex $K_{\bullet}(X,d_X)$ has been constructed, and then it focuses on a small component of the persistence diagrams. Namely, rather than computing the full persistence diagrams (which provide a full invariant of the persistent homology up to isomorphism~\cite{MemoliSinghal19}), the quantum TDA algorithm is aimed at computing the persistent Betti numbers, which are the dimensions of the persistent homology groups~\cite{edelsbrunner2002topological}, and finding the eigenvalues and eigenvectors of the combinatorial Laplacian (obtained from the graph associated with the simplicial complex). The quantum TDA algorithm takes as its starting point the simplicial complex, rather than the raw data. Briefly, if $S$ denotes a simplicial complex constructed from a data set with $n$ elements for some distance threshold $\epsilon$, then each $k$-simplex $s\in S$ is uniquely determined by its $k+1$ vertices. Supposing that there is a total ordering on the vertices, $v_1, v_2,\dots, v_n$, such a simplex $s$ can therefore be expressed as $s=\{v_{i_1},\dots,v_{i_{k+1}}\}$, where $i_1<\cdots<i_{k+1}$ and the distinct vertices $v_{i_1},\dots,v_{i_{k+1}}$ define the simplex $s$. Therefore, the quantum state $\ket{s}$ associated with the simplex $s$ would be the $n$-qubit state with the $i_{1},\dots,i_{k+1}$ qubits in state $\ket{1}$ and all other qubits in state $\ket{0}$ (this is an example of \emph{bit-encoding}, which is discussed in more detail in Appendix~\ref{sec:nattransf}). For example, for $n=7$ and $k=3$, the 2-simplex $s=\{v_2,v_4,v_5\}$ would be represented by the state $\ket{s}=\ket{0101100}$. 

However, one could ask if it is possible to begin the quantum algorithm at the beginning of the pipeline by inputting the raw classical data directly onto the quantum computer and applying quantum algorithms to calculate persistence diagrams from the encoded quantum states. Although Ref.~\cite{LloydGarneroneZanardi16} suggested that this step could be done using amplitude encoding, Ref.~\cite{vlasic2023qtda} showed that amplitude encoding leads to persistence diagrams that differ quite drastically from the persistence diagrams of the original data set. The reason for this stems from the fact that amplitude encoding distorts the distances between the data points so much (as discussed in Example~\ref{ex:amplitudeencoding}) that the persistence diagrams obtained from the quantum encoded states no longer resemble the persistence diagrams of the original data set. The study of how much a persistence diagram is distorted under different mappings is known as \emph{stability} in the (classical) TDA literature, and its study leads to the following notions of structure-preserving map between finite metric spaces.

\begin{definition}
A \define{distance nonincreasing} function $f:(\mathcal{X},d_\mathcal{X})\to (\mathcal{Y},d_\mathcal{Y})$ from one metric space $(\mathcal{X},d_\mathcal{X})$ to another $(\mathcal{Y},d_\mathcal{Y})$ is a function $f:\mathcal{X}\to \mathcal{Y}$ that satisfies $d_{\mathcal{Y}}(f(x_{1}),f(x_{2}))\le d_{\mathcal{X}}(x_{1},x_{2})$ for all $x_1,x_2\in \mathcal{X}$.
An \define{embedding} from $(\mathcal{X},d_\mathcal{X})$ to $(\mathcal{Y},d_\mathcal{Y})$ is a function $f:\mathcal{X}\to\mathcal{Y}$ that satisfies $d_{\mathcal{Y}}(f(x_{1}),f(x_{2}))=d_{\mathcal{X}}(x_{1},x_{2})$ for all $x_1,x_2\in \mathcal{X}$, i.e., $f$ is \define{distance-preserving}. 
\end{definition}

From the categorical perspective that we will introduce soon, quantum encodings that preserve metric structure in either of these ways are said to be \emph{morphisms} in a particular \emph{category}. Additionally, when two mathematical objects are viewed as being equivalent in a way that respects structure via such morphisms that go back and forth between those two objects, then those two objects are said to be \emph{isomorphic}. The following theorem identifies how preserving the metric during the encoding stage guarantees that the constructions in the persistent homology pipeline are isomorphic, whether they are built from the classical data or the quantum encoded data. 

\begin{theorem}
\label{thm:strictstability}
Let $(\mathcal{X},d_\mathcal{X})$ and $(\mathcal{Y},d_\mathcal{Y})$ be two metric spaces, let $X\subseteq\mathcal{X}$ be a finite subset equipped with the induced metric $d_{X}$ from $d_{\mathcal{X}}$. If $f:(X,d_{X})\to(\mathcal{Y},d_{\mathcal{Y}})$ is an embedding, let $Y:=f(X)$ be the image of $X$ under $f$ equipped with the induced metric $d_{Y}$ from $d_{\mathcal{Y}}$. Then the filtered simplicial complexes $K_{\bullet}(X,d_X)$ and $K_{\bullet}(Y,d_Y)$ are isomorphic, the persistence vector spaces $H_{\bullet}(X,d_X)$ and $H_{\bullet}(Y,d_Y)$ are isomorphic, and also the persistence diagrams $\mathrm{dgm}_{\bullet}(X,d_{X})$ and $\mathrm{dgm}_{\bullet}(Y,d_{Y})$ are all isomorphic.
\end{theorem}

For brevity, we refer the reader to Refs.~\cite{CarlssonTDA,CSEHstability2007,MemoliSinghal19} for more details on the definitions and proof, the latter of which follows from the fact that each of the arrows in the TDA pipeline in Figure~\ref{fig:TDApipeline} defines a \emph{functor}, a concept that we will define in Section~\ref{sec:functbg}. Instead, we will elaborate on the meaning of Theorem~\ref{thm:strictstability} in the context of quantum encodings for TDA. The main point is that if the initial quantum encoding preserves the distances exactly with respect to some suitably chosen metric on the set of quantum states, then the associated persistent homology of the quantum states obtained from encoding the classical data must agree with the persistent homology of the original data. Because different types of metrics exist on the space of quantum states (some examples will be given in the next section), one must choose both a suitable metric on the space of quantum states as well as a suitable encoding that preserves the distance from the classical data to the quantum states. 

However, if one does not use a distance-preserving quantum encoding, then a more general stability theorem dictates how the persistence homologies and diagrams may change~\cite{CSEHstability2007,MemoliSinghal19}. Namely, there is a bound relating the \emph{Gromov--Hausdorff distance} between the original metric space of classical data and the metric space of encoded quantum data to the \emph{bottleneck distance} between the associated persistence diagrams using distance nonincreasing maps. It remains an open problem to explicitly construct a distance nonincreasing quantum encoding that has a persistence diagram close to the one obtained from the original classical data.
For brevity, we will not elaborate on the details of this here. Instead, we focus on another setting in which structure could be preserved in the context of quantum metric learning.

\subsection{Distance for metric learning}
\label{ssec:dist_metriclearning}

In the previous section, we saw how topological properties of data can be preserved as long as one uses encodings that are distance-preserving or, more generally, distance nonincreasing. 
However, this is not always possible, since an encoding might drastically distort distances. 
For example, faces of an individual from slightly different angles might be considered similar by us, while the vectors of pixels might be widely separated. A similar situation occurs when encoding classical data into a quantum system. 
The idea of \emph{metric learning} is to encode the notions of similarity, dissimilarity, or relative constraints between data points in order to recover a metric that accurately describes the proximity of data~\cite{Bellet2015metric,Xingetal2002,Weinberger2009distance,Chopra2005,Kulis13,KaBi19}. This is most practically done when the metric depends on a reasonable set of parameters that can be optimized, such as in linear metric learning (cf.\ Appendix~\ref{app:CLMLMM}). 
In most of these situations, recovering a metric employs the mathematical concept of an embedding, that is, a distance-preserving function, between metric spaces, as described in the previous section.

It follows from the definition of a metric space (cf.\ Definition~\ref{defn:metrispace}) that an embedding is automatically injective as a function. Moreover, any injective function from an arbitrary set $\mathcal{X}$ into a metric space $(\mathcal{Y},d_{\mathcal{Y}})$ induces a metric on $\mathcal{X}$, as we will recall in Lemma \ref{lem:inducedmetric} below. This is an important idea in the context of metric learning because one might not know which metric to use on $\mathcal{X}$, and the space of metrics is overwhelmingly vast and unsuitable for optimization in most cases. Instead, one parametrizes the metric through other means~\cite{Bellet2015metric}. Namely, one fixes a suitable metric space $(\mathcal{Y},d_{\mathcal{Y}})$, such as Euclidean space in the case of classical machine learning or the state space of a quantum system in the case of quantum machine learning (where the choice of metric in the quantum case depends on many factors)~\cite{Kulis13,lloyd20}. What varies, then, is the encoding map $f:\mathcal{X}\to\mathcal{Y}$, which is assumed to be an injective function. One chooses $f$ according to some class of models, and the following mathematical fact allows us to define a metric $d_{\mathcal{X}}$ on $\mathcal{X}$ from the metric $d_{\mathcal{Y}}$ on $\mathcal{Y}$ and the encoding $f$. 

\begin{lemma}
\label{lem:inducedmetric}
Let $\mathcal{X}$ be a set, $(\mathcal{Y},d_{\mathcal{Y}})$ a metric space,  $\mathcal{X}\xrightarrow{f}\mathcal{Y}$ a function, and $d_{\mathcal{X}}:\mathcal{X}\times \mathcal{X}\to[0,\infty)$ the function 
\begin{equation}
\label{eq:dXpullbackmetric}
d_{\mathcal{X}}(x_1,x_2):=d_{\mathcal{Y}}\big(f(x_1),f(x_2)\big)
\end{equation}
defined for all $x_1,x_2\in \mathcal{X}$. 
Then, 
$d_{\mathcal{X}}$ is a metric on $\mathcal{X}$ if and only if $f$ is one-to-one. In such a case, $(\mathcal{X},d_{\mathcal{X}})\xrightarrow{f}(\mathcal{Y},d_{\mathcal{Y}})$ is an embedding. 
\end{lemma}

The method of obtaining a metric on $\mathcal{X}$ in this manner via some injective function $f$ is called \emph{global} (possibly nonlinear) metric learning~\cite{Chopra2005,Kedem2012,Kulis13}. In general, the parameters describing such a metric can come from two sources: (a) a family of metrics on the codomain $\mathcal{Y}$ (such as the Minkowski metrics induced by the $L_{p}$ norm~\cite{Bellet2015metric}) and (b) a family of injective maps $f:\mathcal{X}\to\mathcal{Y}$ (such as a space of injective linear transformations in linear metric learning, as described in Appendix~\ref{app:CLMLMM}). The proof of Lemma~\ref{lem:inducedmetric} is given in Appendix~\ref{app:metriclearningsms}. 
The lemma itself motivates the following definition.

\begin{definition}
Given a set $\mathcal{X}$, a metric space $(\mathcal{Y},d_{\mathcal{Y}})$, and a one-to-one function $\mathcal{X}\xrightarrow{f}\mathcal{Y}$, the metric $d_{\mathcal{X}}$ on $\mathcal{X}$ constructed in Lemma~\ref{lem:inducedmetric} is called the \define{pull-back metric} or the \define{embedding metric}.
\end{definition}

We can now use this idea to make rigorous sense of quantum metric learning~\cite{lloyd20}. In this setting, the codomain $\mathcal{Y}$ is taken to be $\states(\mathcal{H})$, the convex space of states on a Hilbert space $\mathcal{H}$. Utilizing Lemma~\ref{lem:inducedmetric}, one uses a quantum encoding map $\rho:\mathcal{X}\to\states(\mathcal{H})$ together with a metric on $\states(\mathcal{H})$ to define the distances between data points in $\mathcal{X}$ via the embedding metric $d_{\mathcal{X}}$. Intuitively, one hopes to arrive at a metric $d_{\mathcal{X}}$ such that $(\mathcal{X},d_{\mathcal{X}})$ is close to $(\mathcal{X},d_{\mathcal{X}_{\mathrm{true}}})$, where $d_{\mathcal{X}_{\mathrm{true}}}$ is some idealized but unknown metric and where closeness can be defined with respect to the Gromov--Hausdorff distance, for example~\cite{BuBuIv01}.  
There are many options for the models used for the quantum encoding, including those from Appendix~\ref{app:onepgroupqencode}. 
Moreover, there are several options for metrics on the space of states $\states(\mathcal{H})$. 
The following definition provides some possibilities of metrics on quantum states~\cite{Hillerydistance1987,Dieks1983distance,BraunsteinCaves1994,Bures1969,Uhlmann1976,Hubner1992}. 

\begin{definition}
\label{defn:quantumdistances}
Fix a Hilbert space $\mathcal{H}$ and let $\rho$ and $\sigma$ be two density matrices on $\mathcal{H}$. 
The \define{trace/$\ell_{1}$ distance} between $\rho$ and $\sigma$ is 
\begin{equation}
d_{\Tr}(\rho,\sigma)=
\lVert\rho-\sigma\rVert_{1}\equiv
\Tr\left[\sqrt{(\rho-\sigma)^{\dag}(\rho-\sigma)}\right].
\end{equation}
The \define{Hilbert--Schmidt/Frobenius/$\ell_{2}$ distance} is
\begin{equation}
d_{HS}(\rho,\sigma)=\sqrt{\Tr\left[(\rho-\sigma)^{\dag}(\rho-\sigma)\right]}.
\end{equation}
The \define{Bures/Helstrom/infidelity distance} is
\begin{equation}
d_{B}(\rho,\sigma)=\sqrt{2\left(1-\sqrt{F(\rho,\sigma)}\right)},
\end{equation}
where
\begin{equation}
F(\rho,\sigma)=\left(\Tr\left[\sqrt{\sqrt{\rho}\sigma\sqrt{\rho}}\right]\right)^{2}
\end{equation}
denotes the \define{fidelity} between the states $\rho,\sigma\in\mathcal{S}(\mathcal{H})$. 
\end{definition}

These distance functions (as well as scalar multiples of them) provide several examples of metrics on the space $\states(\mathcal{H})$ of quantum states. There are several other interesting possibilities, which can be found in Refs.~\cite{Kuzmak2021,watrous2018theory,ZyczkowskiSlomczynski1998,Dodonovetal1999} and the references therein. 
Let $d_{\mathcal{H}}$ be the notation for any one of these metrics. The goal of (classical) \emph{global metric learning} is then to find an injective \emph{function} $\rho:\mathcal{X}\to\states(\mathcal{H})$ from which we can pull back the metric $d_{\mathcal{H}}$ to the data domain $\mathcal{X}$. One may additionally add on a parameter space $\Theta$ to define a function of the form $\rho:\mathcal{X}\times\Theta\to\states(\mathcal{H})$ from which a specific parameter $\theta\in\Theta$ is chosen to optimize some machine learning task. For example, this is essentially what is done in Ref.~\cite{lloyd20}. 

However, there are a few important remarks to be made about quantum metric learning as used in Ref.~\cite{lloyd20}. First, it is in fact a special case of global nonlinear metric learning where the space for the encoded data is the state space of a quantum system (as opposed to a standard Euclidean space $\R^{d}$, for example)~\cite{Kulis13,Bellet2015metric}.  
Second, many of the encodings used in Ref.~\cite{lloyd20} are technically not embeddings in the sense defined here. This is because the quantum encodings are not always \emph{injective} functions. We will soon illustrate this in an example, which shows that the distances between \emph{different} points can vanish. This is important because then Lemma~\ref{lem:inducedmetric} fails, and we need a modification. This is obtained by the notion of a \emph{semi-metric space} (also called a \emph{fuzzy metric space} and \emph{pseudo-metric space})~\cite{Spivak09,BuBuIv01,Bellet2015metric}, whose importance was also recognized in the development of UMAP~\cite{McInnesUMAP}. 

\begin{definition}
\label{defn:semimetricspace}
A \define{semi-metric space} is a pair $(\mathcal{X},d_{\mathcal{X}})$ as in the definition of a metric space, but $d_{\mathcal{X}}$ satisfies the same properties in Definition~\ref{defn:metrispace} except that $d_{\mathcal{X}}(x_1,x_2)=0$ no longer implies $x_1=x_2$, i.e., there may exist distinct points whose distance is zero. 
\end{definition}

The importance of this definition is to allow for the flexibility of quantum encodings that are not injective. Although the sampled \emph{data} should perhaps be embedded in a one-to-one fashion, it may be that the larger data domain might not strictly embed under the encoding. This is illustrated in Figure~\ref{fig:semimetric1ddataset}, where we have reproduced a version of Figure~4 of Ref.~\cite{lloyd20}. The next lemma shows that a semi-metric space structure on the domain of an encoding can be obtained even when the encoding is not necessarily injective.

\begin{figure*} 
\begin{tabular}{c}
    \begin{tabular}{ccc}
    \begin{tikzpicture}[scale=1.25]
    \node at (0,1.5) {};
    \node at (0,-1.5) {};
    \node at (0,-1.75) {(a)};
    \draw[->] (-2.15,0) -- (2.15,0) node[right]{$x$};
    \foreach \t in {-2,-1,0,1,2}{
      \draw[-] (\t,0.05) -- (\t,-0.05) node[below]{$\t$};
      }
    \foreach \x in {-0.168069, 0.0294028, 0.0919156, 0.562233, -0.31038, 0.462528, -0.636511, 0.176429, 0.0919156, -0.0585628, -0.0656013, -0.0656013, 0.260227, -0.168302, 0.427166} {
      \node[orange] at (\x,0) {$\boldsymbol{\circ}$};
      }
    \foreach \x in {1.79, 1.91651, 1.529, -1.65776, 1.40029, 1.58424, -1.83566, 1.63235, -1.53134, -1.33236, -1.1282, 1.85109, -1.33236, 1.53148, 1.53148} {
      \node[blue] at (\x,0) {$\boldsymbol{\times}$};
      }
    \end{tikzpicture}
    \qquad
    &
    \qquad
\includegraphics{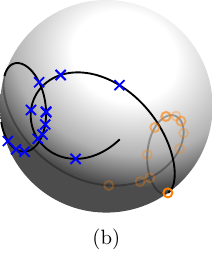}
    \qquad
    & 
    \qquad
\includegraphics{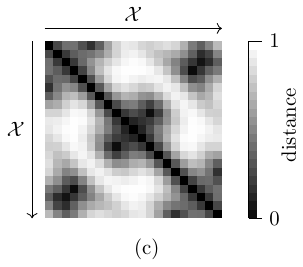}
    \end{tabular}
\\
\begin{tikzpicture}
\node at (0,0){%
\begin{quantikz}
\setwiretype{n} & \lstick{\ket{0}} & \setwiretype{q} & \gate{R_{X}(x)} & \gate{R_{Y}(\theta_1)} &  \gate{R_{X}(x)} & \gate{R_{Y}(\theta_2)} &  \gate{R_{X}(x)} & \gate{R_{Y}(\theta_3)} &  \gate{R_{X}(x)} & \gate{R_{Y}(\theta_4)} &  \gate{R_{X}(x)} & 
\end{quantikz}
};
\node at (0,0.70) {};
\node at (0,-0.70) {(d)};
\end{tikzpicture}
\end{tabular}
\caption{(a) Two classes of data on $\mathcal{X}=[-2,2]$ are presented. One class has orange data points labeled {\color{orange}$\boldsymbol{\circ}$} clustered near $0$ and the other class has blue data points labeled {\color{blue}$\boldsymbol{\times}$} clustered near $2$ and $-2$.
(b) The image of the quantum encoding $\rho:[-2,2]\to\states(\C^{2})$ with $\states(\C^{2})$ visualized as the Bloch ball and the encoding $\rho$ defined in~\eqref{eq:rhothetaxLloyd} using the optimized parameters 
$\theta_1=0.31$, $\theta_2=1.48$, $\theta_3=0$, and $\theta_4=0$.
(c) A visualization of the pull-back metric from the quantum encoding $\rho:\mathcal{X}\to\states(\C^{2})$ back onto the domain $\mathcal{X}$. The metric has been uniformly rescaled so that the maximum distance between points is $1$.
(d) The quantum circuit architecture from Ref.~\cite{lloyd20} that we used to achieve the quantum encoding shown in (b). 
}
\label{fig:semimetric1ddataset}
\end{figure*}
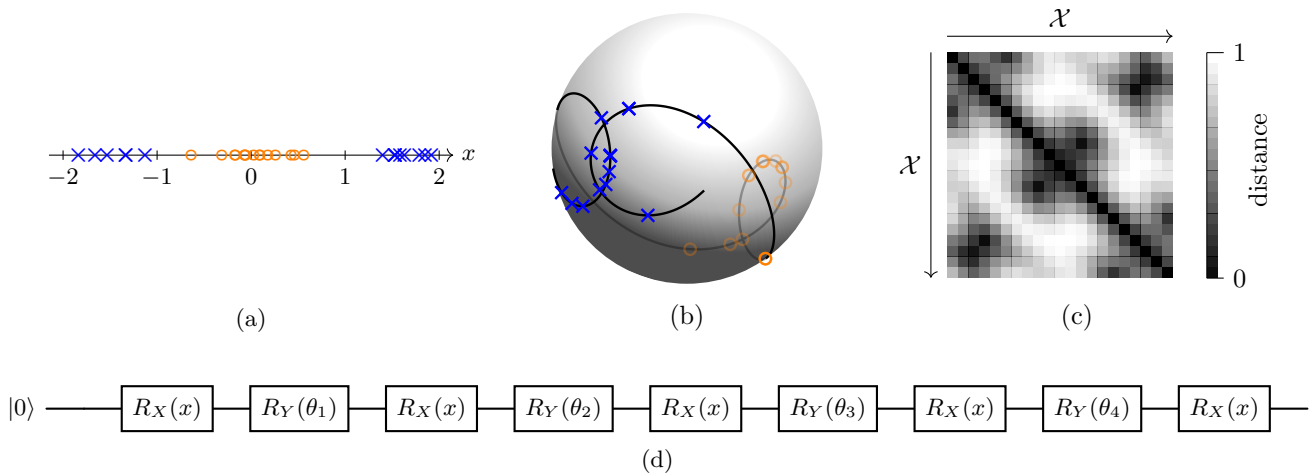

\begin{lemma}
\label{lem:inducedsemimetric}
Let $\mathcal{X}$ be a set, $(\mathcal{Y},d_{\mathcal{Y}})$ a metric space,  $\mathcal{X}\xrightarrow{f}\mathcal{Y}$ a function, and $d_{\mathcal{X}}:\mathcal{X}\times \mathcal{X}\to[0,\infty)$ the function from~\eqref{eq:dXpullbackmetric}. Then $d_{\mathcal{X}}$ is a semi-metric on $\mathcal{X}$. 
\end{lemma}

\begin{proof}
The same argument in the proof of Lemma~\ref{lem:inducedmetric} applies with the exception of the last paragraph. 
\end{proof}

We leave the details behind Figure~\ref{fig:semimetric1ddataset} to Appendix~\ref{app:metriclearningsms} so that we may first summarize our examples and then provide the general categorical framework that encompasses them.

\subsection{Summarizing our examples}
\label{sec:examplesummary}

As we conclude this brief survey of examples, observe how each paradigm in Sections \ref{sec:symmetry}, \ref{sec:top_smooth}, \ref{ssec:TDA}, and \ref{ssec:dist_metriclearning} involved the choice of a mathematical structure that is paired with the data domain $\mathcal{X}$ and quantum state space $\states(H)$. Finding a structure-preserving quantum encoding is thus a multi-step process. One first selects a structure of interest, then one equips the data domain and set of states with that structure, and finally one looks for a function $\rho: \mathcal{X}\to \states(H)$ with the property that it preserves the structure. 

Although relatively straightforward, the perspective put forth in this article is that category theory provides a natural language in which to recast these design steps, potentially lending additional mathematical clarity and insight. Said explicitly, a categorical perspective reveals that designing a structure-preserving quantum encoding amounts to looking for mathematical objects $\mathcal{X}'$ and $\states(\mathcal{H})'$ in a particular \emph{category} that are, in fact, $\mathcal{X}$ and $\states(\mathcal{H})$ \emph{equipped with extra structure}, an idea that can be formalized by invoking the categorical concept of a \emph{functor}. One then requires that the desired encoding $\rho:\mathcal{X} \to \states(\mathcal{H})$, which is a priori blind to any structure, corresponds to a \emph{morphism} $\rho':\mathcal{X}' \to \states(\mathcal{H})'$ in that category. This new mapping $\rho'$ can be thought of as an improved version of $\rho$ that not only \textit{sees} the mathematical structure but also \textit{respects} it. This, too, can be made precise in the language of category theory. Our goal in the next section is to provide the reader with the categorical tools needed to grasp this perspective.

\section{Category theory captures structure}
\label{sec:CatsQE}

Informally speaking, a category consists of objects and relationships, called morphisms, between them. For example, there is a category of all topological spaces and continuous functions between them as well as a category of all $G$-sets and $G$-equivariant maps between them for any group $G$. Definitions and examples will be presented more formally below, but as mentioned above, categories can be used to describe data domains and quantum state spaces together with any additional structure that may be imposed on them. As already hinted, examples of such structures include symmetries, topologies, smooth structures for differentiability, metrics, linear algebraic structure such as convex combinations, and more. Meanwhile, a structure-preserving quantum encoding is a morphism in one (or more) of these categories. 

In Section \ref{sec:catsbg}, we will present the formal definition of a category with examples that appear in classical and quantum machine learning, focusing particularly on the examples listed throughout Section \ref{sec:SPQE}. As hinted previously in Section \ref{sec:examplesummary}, our main perspective of structured quantum encodings requires the concept of a functor (specifically, a \emph{forgetful functor}), which is a particular kind of passage between categories. As such, functors will be described in Section \ref{sec:functbg}, and many relevant examples will be provided. 
With this background in hand, Section~\ref{sec:lifting} will then reveal the main perspective proposed in this paper, namely that \textit{reducing the search space from general quantum encodings to structure-preserving quantum encodings is achieved by lifting the problem from a structure-less category to one containing more structure in the presence of some forgetful functor.} 

\subsection{Categories}
\label{sec:catsbg}

Before presenting the formal definition of a category~\cite{mac2013categories}, which may appear abstract at first glance, it is helpful to know that a classic example of a category involves sets and functions between them --- including quantum encodings, one of the most basic ingredients in QML applications.  
Let us walk through this example slowly, making some simple observations about functions along the way. To start, observe that for any two sets $\mathcal{X}$ and $\mathcal{Y}$, there is a set of functions $f:\mathcal{X}\to \mathcal{Y}$ from $\mathcal{X}$ to $\mathcal{Y}$. Functions, though quite minimalistic, nevertheless exhibit some structure when viewed as a whole. Namely, for any two functions $f: \mathcal{X}\to \mathcal{Y}$ and 
$g: \mathcal{Y}\to \mathcal{Z}$, their composite is also a function $g\circ f:\mathcal{X}\to \mathcal{Z}$. Moreover, this composition is associative in the sense that given three functions $f:\mathcal{X}\to \mathcal{Y}$ and $g:\mathcal{Y}\to \mathcal{Z}$ and $h:\mathcal{Z}\to \mathcal{W}$, the composites $h\circ (g\circ f)$ and $(h\circ g)\circ f$ are equal. Finally, for each set $\mathcal{X}$, there is a particular function $\id_\mathcal{X}:\mathcal{X}\to \mathcal{X}$ defined by $\id_\mathcal{X}(x)=x$ for all $x\in \mathcal{X}$ that acts as a unit/identity for this composition since $\id_{\mathcal{Y}}\circ f=f\circ \id_{\mathcal{X}}=f$ for all functions $f:\mathcal{X}\to \mathcal{Y}$. 
These observations show us that we have some mathematical ``objects'' (sets) and ``arrows'' or ``morphisms'' (functions) between them that interact together in a reasonable way via the composition rule. This is a classic example of a category and is denoted by $\mathbf{Set}$. In this way, a data domain $\mathcal{X}$ and the set $\states(\mathcal{H})$ of states of a quantum system are objects in the category $\mathbf{Set}$. Furthermore, a quantum state encoding $\rho:\mathcal{X}\to\states(\mathcal{H})$ as presented in Definition \ref{def:encoding} is a morphism in $\mathbf{Set}$. 

The formal definition of a category abstracts the example of sets and functions in the following way.

\begin{definition}\label{def:category}
    A category $\mathbf{C}$ consists of
    \begin{itemize}
    \item a collection of \define{objects}, denoted as $\mathcal{X},\mathcal{Y},\ldots$
    \item for every pair of objects $\mathcal{X},\mathcal{Y}$, a set of \define{morphisms} from $\mathcal{X}$ to $\mathcal{Y}$, depicted as arrows $f:\mathcal{X}\to \mathcal{Y}$
    \item a \define{composition rule}: for any morphisms $f~:~\mathcal{X}\to~\mathcal{Y}$ and $g:\mathcal{Y}\to \mathcal{Z}$, a specified morphism $g\circ f:\mathcal{X}\to~\mathcal{Z}$. 
    \end{itemize}
    Moreover, these items must satisfy the following axioms:
    \begin{itemize}
    \item Composition is \define{associative}; that is, for any three composable morphisms $f:\mathcal{X}\to \mathcal{Y}$, $g:\mathcal{Y}\to \mathcal{Z}$, $h: \mathcal{Z}\to \mathcal{W}$, the composites $h\circ (g\circ f)$ and $(h\circ g)\circ f$ are equal.
    \item There exist \define{identity morphisms}; that is, for every object $\mathcal{X}$ there exists a morphism $\id_\mathcal{X}:\mathcal{X}\to \mathcal{X}$ such that $f\circ \id_\mathcal{X}=\id_{\mathcal{Y}}\circ f=f$ for all morphisms $f:\mathcal{X}\to \mathcal{Y}$.
    \end{itemize}
\end{definition}

Below are three additional examples of categories that have already appeared in our analysis of quantum encodings. 

\begin{example}
[$G$-Sets and $G$-Equivariant Functions]
\label{ex:gsets}
For a fixed group $G$, there is a category $G\text{-}\mathbf{Set}$ whose objects are \emph{$G$}-sets, where a morphism between $G$-sets is a $G$-equivariant function. The following observations show that this is indeed a category:
    \begin{itemize} 
    \item The composite of two $G$-equivariant functions is again $G$-equivariant.
    \item Composition is associative due to the associativity of function composition.
    \item For each $G$-set $(\mathcal{X},\alpha)$, the identity function $\id_\mathcal{X}$ on $\mathcal{X}$ is $G$-equivariant.
    \end{itemize}
We have already seen a class of examples coming from geometric quantum machine learning~\cite{larocca2022group,ragone2023representation,meyer2023exploiting}, namely, $G$-equivariant encodings $\rho:\mathcal{X}\to\states(\mathcal{H})$, where $G$ acts on the data domain $\mathcal{X}$ by some action $\alpha$ and $G$ acts on $\states(\mathcal{H})$ by the adjoint action of some unitary representation of $G$ on $\mathcal{H}$ (recall Section \ref{sec:symmetry}). Thus, such a $G$-equivariant encoding $\rho$ is a morphism in the category $G\text{-}\mathbf{Set}$. Moreover, since the composite of equivariant functions is still an equivariant function, this is what enables one to build equivariant layers in quantum circuits and still obtain an overall equivariant function.
\end{example}

\begin{example}
[Topological Spaces and Continuous Functions] 
\label{ex:top}
There is a category $\mathbf{Top}$ whose objects are topological spaces, where a morphism between spaces is a continuous function. The following observations show that this is indeed a category:
    \begin{itemize}
    \item The composite of two continuous functions is again continuous.
    \item Composition is associative due to the associativity of function composition.
    \item For each topological space $(\mathcal{X},\tau_\mathcal{X})$, the identity function $\id_\mathcal{X}$ on $\mathcal{X}$ is continuous.
    \end{itemize}
In Section \ref{sec:top_smooth}, we saw that angle encodings and amplitude encodings are two examples of morphisms in the category $\mathbf{Top}$ when the data domain $\mathcal{X}$ (being Euclidean space) and the set of states $\states(\mathcal{H})$ are given their standard topologies~\cite{BeZy17,LeeISM13,Mu00}. These two examples of encodings are also morphisms in the category consisting of smooth manifolds and smooth maps between them.
\end{example}

\begin{example}[Metric Spaces and Embeddings]
\label{ex:Metemb}
There is a category $\mathbf{Met}^{\mathrm{emb}}$ whose objects are metric spaces, where a morphism between metric spaces is an embedding. The following observations show that this is indeed a category:
    \begin{itemize}
    \item The composite of two embeddings is an embedding.
    \item Composition is associative due to the associativity of function composition.
    \item For each metric space $(\mathcal{X},d_\mathcal{X})$, the identity function $\id_\mathcal{X}$ on $\mathcal{X}$ is an embedding.
    \end{itemize}
The category $\mathbf{Met}^{\mathrm{emb}}$ appeared in Section~\ref{ssec:TDA} for TDA and Section~\ref{ssec:dist_metriclearning} for quantum metric learning, where one equips the set of states $\states(\mathcal{H})$ with a metric and uses a quantum state encoding $\rho:\mathcal{X}\to\states(\mathcal{H})$ to define a metric on the data domain $\mathcal{X}$.
\end{example}

Our discussion surrounding TDA in Section~\ref{ssec:TDA} likewise involved metric spaces but additionally included more general morphisms between them. Namely, we considered distance nonincreasing functions rather than embeddings. More generally, one might wish to adjust the extent to which a function preserves distances, which can lead to different categories having the same objects but different morphisms. The following example summarizes this observation.

\begin{example}
[Metric Spaces and Distance Nonincreasing Functions]
\label{ex:MetNim}
There is a category $\mathbf{Met}^{\mathrm{nif}}$ whose objects are metric spaces $(\mathcal{X},d_\mathcal{X})$. A morphism in $\mathbf{Met}^{\mathrm{nif}}$ between metric spaces is a distance nonincreasing function.   
As before, this is a category since the following criteria are satisfied:
    \begin{itemize}
    \item The composite of two distance nonincreasing functions is a distance nonincreasing function.
    \item Composition is associative due to the associativity of function composition.
    \item For each metric space $(\mathcal{X},d_\mathcal{X})$, the identity function $\id_\mathcal{X}$ on $\mathcal{X}$ is distance nonincreasing.
    \end{itemize}
Notice that embeddings between metric spaces are also distance nonincreasing. For this reason, $\mathbf{Met}^{\mathrm{nif}}$ is called a \emph{subcategory} of $\mathbf{Met}^{\mathrm{emb}}$, meaning that all objects and morphisms of $\mathbf{Met}^{\mathrm{emb}}$ are also objects and morphisms of $\mathbf{Met}^{\mathrm{nif}}$, and the identities and composition rules agree on these common objects and morphisms. 
\end{example}

Each of the four examples presented here describe categories whose objects are sets with structure (a group action, a topology, a metric) and whose morphisms are functions that preserve (part of) that structure. We have also seen that some quantum encodings, when chosen appropriately, can be viewed as morphisms in these categories. A priori, however, a quantum encoding is merely a function between sets, as was introduced in Definition \ref{def:encoding}. In QML, therefore, one sometimes works within the category $\mathbf{Set}$, while at other times one works in a different category, such as those listed in Examples \ref{ex:gsets}, \ref{ex:top}, \ref{ex:Metemb}, or \ref{ex:MetNim}. So, to take account of when a passage is being made from one category to another, we must understand mappings between categories. This leads to the concept of a functor.

\subsection{Functors}
\label{sec:functbg}

While categories are useful for organizing mathematical objects and relationships (morphisms) between them, much of mathematics concerns relationships between \textit{categories} themselves. 
This involves assigning the objects and morphisms of one category to another in a manner that respects morphism composition and identities. 

To situate this in the context of QML, think back to $G$-equivariant quantum state encodings for some group $G$. Given a unitary representation $V:G\to\unitary(\mathcal{H})$ of $G$, one obtains a $G$-equivariant function $\rho: (\mathcal{X},\alpha)\to(\states(\mathcal{H}),\mathrm{Ad}_V)$ which, as we have seen, is a morphism in the category $G\text{-}\mathbf{Set}$. But there is an obvious way to pass from $G\text{-}\mathbf{Set}$ to the category $\mathbf{Set}$: simply discard, or forget, the group actions associated with each set (cf. Figure~\ref{fig:Gsetforget}), and assign each $G$-equivariant function to the function itself.
Although straightforward, let us denote this assignment by $F:G\text{-}\mathbf{Set}\to\mathbf{Set}$, so that for an arbitrary $G$-set $(\mathcal{X},\alpha)$, we have $F(\mathcal{X},\alpha)=\mathcal{X}$ and for any $G$-equivariant function $f$, we have $F(f)=f$. Then, $F$ behaves well with respect to the two categories in the following ways.

\begin{figure}
    \begin{tikzpicture}
    \node at (0,0) {
    $
\xy 0;/r.25pc/:
    (-15,20)*+{G\text{-}\mathbf{Set}}="SS";
    (15,20)*+{\mathbf{Set}}="S";
    {\ar"SS";"S"^{\text{forget action}}};
    (-15,8)*+{(\mathcal{X},\alpha)}="SX";
    (-40,8)*+{\text{data domain}};
    (15,8)*+{\mathcal{X}}="X";
    (-15,0)*+{(\states(\mathcal{H}),\Ad_{V})}="SH";
    (-40,0)*+{\text{state space}};
    (15,0)*+{\states(\mathcal{H})}="H";
    (-15,-8)*+{(\mathcal{Y},\beta)}="SY";
    (-40,-8)*+{\text{arbitrary $G$-set}};
    (15,-8)*+{\mathcal{Y}}="Y";
    {\ar@{|->}"SH";"H"^(0.59){\text{forget $\Ad_{V}$}}};
    {\ar@{|->}"SX";"X"^{\text{forget $\alpha$}}};
    {\ar@{|->}"SY";"Y"^{\text{forget $\beta$}}};
\endxy
   $
   };
   \end{tikzpicture}
   \caption{Forgetting the action on a $G$-set and just remembering the underlying set defines a functor from $G\text{-}\mathbf{Set}$ to $\mathbf{Set}$. This functor can be applied to every $G$-set, including data domains and state spaces, regardless of the $G$-actions they come equipped with.}
   \label{fig:Gsetforget}
\end{figure}

First, for any two composable $G$-equivariant functions $f$ and $g$, we have, somewhat trivially,
\begin{equation}
F(g\circ f)=g\circ f = F(g)\circ F(f).
\end{equation}
In words, this means that composing two $G$-equivariant functions and \textit{then} discarding the group actions on their domains and codomains results in the same function obtained by \textit{first} discarding the group actions and then composing the underlying functions. In more concise terms, $F$ preserves the composition rules of the two categories.

Second, for any $G$-set $(\mathcal{X},\alpha)$, we have that
\begin{equation}
F(\id_{(\mathcal{X},\alpha)})=\id_{(\mathcal{X},\alpha)}=\id_\mathcal{X} =\id_{F(\mathcal{X},\alpha)}.
\end{equation}
The first and third equalities follow from the definition of $F$ on morphisms and objects, respectively. The second equality follows from the fact that the identity functions on $\mathcal{X}$, when $\mathcal{X}$ is viewed as a $G$-set or as a (regular) set, are both defined by $x\mapsto x$ for all $x\in\mathcal{X}$. So, forgetting the group action and \textit{then} considering the identity function on the underlying set results in the same function obtained by \textit{first} considering the $G$-equivariant identity function and then forgetting about the group action. Succintly put, $F$ preserves identity morphisms.

In summary then, $F:G\text{-}\mathbf{Set}\to\mathbf{Set}$ is an assignment on the objects and morphisms of one category to another that preserves the categorical structure. This is an example of a functor.

\begin{definition} 
\label{def:functor}
A \define{functor} $F:\mathbf{C}\to\mathbf{D}$ between categories consists of 
\begin{itemize}
    \item an object $F(\mathcal{X})$ in $\mathbf{D}$ for every object $\mathcal{X}$ in $\mathbf{C}$
    \item a morphism $F(f): F(\mathcal{X})\to F(\mathcal{Y})$ in $\mathbf{D}$ for every morphism $f:\mathcal{X}\to \mathcal{Y}$ in $\mathbf{C}$.
\end{itemize}
Moreover, these items must satisfy the following axioms:
\begin{itemize}
    \item \emph{Composition is preserved;} that is, $F(g\circ f)={F(g)\circ F(f)}$ for all  morphisms $f:\mathcal{X}\to \mathcal{Y}$ and $g:\mathcal{Y}\to \mathcal{Z}$ in $\mathbf{C}$.
    \item \emph{Identities are preserved;} that is, $F(\id_\mathcal{X})=\id_{F(\mathcal{X})}$ for every object $\mathcal{X}$ in $\mathbf{C}$.
\end{itemize}
\end{definition}

A few more examples of functors will appear below, but to help connect these upcoming examples to the goals of this article, let us first establish some terminology relating to properties that functors may satisfy. 
Let $\mathbf{C}(\mathcal{X},\mathcal{Y})$ denote the set of all morphisms from object $\mathcal{X}$ to object $\mathcal{Y}$ in a category $\mathbf{C}$. Given another category $\mathbf{D}$, if a functor $F:\mathbf{C}\to\mathbf{D}$ has the property that for all objects $\mathcal{X}$ and $\mathcal{Y}$ in $\mathbf{C}$, the function $\mathbf{C}(\mathcal{X},\mathcal{Y})\to\mathbf{D}(F(\mathcal{X}),F(\mathcal{Y}))$ defined by $f\mapsto F(f)$ is injective, then $F$ is said to be \define{faithful}. If the function $f\mapsto F(f)$ is surjective, then $F$ is said to be \define{full}. If it is a bijection, then $F$ is said to be \define{fully faithful}.  
As the next example illustrates, we have already seen an example of a functor that is faithful but not, however, full.

\begin{example}
\label{ex:forgetfulGset}
The functor $F:G\text{-}\mathbf{Set}\to \mathbf{Set}$ described prior to Definition \ref{def:functor} (which ``forgets'' the group action of $G$-sets and $G$-equivariant functions to obtain the underlying sets and functions between them) is faithful  since two $G$-equivariant functions are equal whenever they are equal as functions on the underlying sets. However, $F$ is not full since there exist functions between the underlying sets of $G$-sets that are not equivariant. The lack of fullness is what allows us to narrow down the large space of quantum encodings to a significantly smaller subset of structure-preserving (in this case, equivariant) encodings. Hence, the smaller space of equivariant encodings generally simplifies the search for optimal encodings. 
\end{example}

One of the major themes throughout this article has been that many mathematical objects consist of sets equipped with extra structure, so the ability to forget or omit that structure certainly extends far beyond group theory. Below are additional examples of functors relating to our discussions of quantum encodings, which are analogous to the one mentioned in Example~\ref{ex:forgetfulGset}.

\begin{example}[From Metric Spaces to Sets] 
\label{ex:forgetfulMettoSet}
There are functors $\mathbf{Met}^{\mathrm{emb}}\to\mathbf{Set}$ and $\mathbf{Met}^{\mathrm{nif}}\to\mathbf{Set}$ that both assign to every metric space $(\mathcal{X},d_\mathcal{X})$ its underlying set of points $\mathcal{X}$, thus forgetting the distance function $d_\mathcal{X}$. Moreover, the functors also both assign every embedding and distance nonincreasing function to itself. This functor is faithful since two embeddings or two distance nonincreasing functions are equal whenever they are equal as functions on the underlying sets. But neither of these functors is full because not every function is distance preserving nor distance nonincreasing. The lack of fullness in this case restricts the set of all functions to those that do preserve distances or are distance nonincreasing. 
\end{example}

\begin{example}[From Topological Spaces to Sets] 
\label{ex:forgetfulToptoSet}
There is a functor $\mathbf{Top}\to\mathbf{Set}$ that assigns to every topological space $(\mathcal{X},\tau_\mathcal{X})$ its underlying set $\mathcal{X}$ of points, thus forgetting the topology $\tau_\mathcal{X}$. Moreover, the functor assigns every continuous function to itself. As in the previous examples, this functor is faithful but not full. Again, the lack of fullness is what causes the space of morphisms between objects to become smaller when adding structure in the sense that the set of \emph{continuous} functions between topological spaces is in general smaller than the set of \emph{all} functions between the underlying sets. 
\end{example}

Perhaps not surprisingly, each of the functors listed in Examples~\ref{ex:forgetfulGset}, \ref{ex:forgetfulMettoSet}, and \ref{ex:forgetfulToptoSet}
are commonly referred to as ``forgetful functors.'' Loosely speaking, a functor $\mathbf{C}\to \mathbf{D}$ is called \emph{forgetful} if it drops, or omits, some or all of the structure of the objects in $\mathbf{C}$. Though this is not a completely rigorous definition, the notion of a forgetful functor as used in this article will always be of this kind, so a formal definition of this notion is not needed here. The takeaway, though, is that each of the functors in the previous examples was also \textit{faithful}, and faithfulness can be used to formalize the notion of ``structure'' itself. One may say that objects in a category $\mathbf{C}$ are objects in a category $\mathbf{D}$ \define{equipped with extra structure} if there exists a faithful functor $\mathbf{C}\to\mathbf{D}$~\cite{nlab:structure}. In the special case when $\mathbf{D}=\mathbf{Set}$ is the category of sets, $\mathbf{C}$ is further said to be a \define{concrete category}~\cite{riehl}. All of the examples discussed so far --- $\mathbf{Set}$, $G\text{-}\mathbf{Set}$, $\mathbf{Met}^{\mathrm{emb}}$, $\mathbf{Met}^{\mathrm{nif}}$, and $\mathbf{Top}$
--- are concrete categories, as their objects are always sets $\mathcal{X}$ equipped with extra structure (a group action $\alpha,$ a metric $d_\mathcal{X}$, or a topology $\tau_\mathcal{X}$), 
and their morphisms are functions that preserve that structure. Although not all categories are concrete \footnote{As an example, the real line $\mathbb{R}$ can be viewed as a category whose objects are elements $x,y,\ldots\in\mathbb{R}$ and where there is a morphism $x\to y$ whenever $x\leq y$. This is a category due to the transitivity and reflexivity of the partial order, but it is not concrete since morphisms are not structure-preserving functions.}, we are primarily focusing on those that are.

At last, the reward of introducing this high level of abstraction  is that we now have precise mathematical language to describe the design of structure-preserving quantum encodings.

\subsection{Quantum Encoding from a Categorical Perspective}
\label{sec:lifting}

Let us finally reformulate the QML encoding scheme for states from a categorical perspective. (A similar reformulation exists for unitary quantum encodings.) This idea constitutes the main perspective proposed in this paper. 

\let\origdescription\description
\renewenvironment{description}{
  \setlength{\leftmargini}{0em}
  \origdescription
  \setlength{\itemindent}{0em}
  \setlength{\labelsep}{\textwidth}
}
{\endlist}

\begin{description}
    \item[Design Setup] For a given learning task, start with a data domain $\mathcal{X}$ and choose a Hilbert space $\mathcal{H}$. Identify a concrete category $\mathbf{C}$ and equip $\mathcal{X}$ and the set of states $\states(\mathcal{H})$ with extra structure to obtain objects $\mathcal{X}'$ and $\states(\mathcal{H})'$ in $\mathbf{C}$ such that $F(\mathcal{X}')=\mathcal{X}$ and $F(\states(\mathcal{H})')=\states(\mathcal{H})$ are in the image of a faithful forgetful functor $F:\mathbf{C}\to\mathbf{Set}$.
    \item[Design Goal] Find a function $\rho:\mathcal{X}\to \states(\mathcal{H})$ with the property that there exists a morphism $\rho': \mathcal{X}'\to \states(\mathcal{H})'$ in $\mathbf{C}$ such that $F(\rho')=\rho$. The morphism $\rho'$ is the desired, structure-preserving quantum encoding.
\end{description}

The process described in the Design Goal can be illustrated diagrammatically. Among all of the set-theoretic functions $\rho:\mathcal{X}\to\states(\mathcal{H})$ in $\mathbf{Set}$, there might exist a few that can be ``lifted'' to a morphism $\rho':\mathcal{X}'\to\states(\mathcal{H})'$ in $\mathbf{C}$, the category of structure-preserving morphisms, as in the following diagram.

\[
\xy 0;/r.25pc/:
    (-30,7.5)*+{\mathbf{C}}="SS";
    (-30,-7.5)*+{\mathbf{Set}}="S";
    {\ar"SS";"S"_{F}};
    (0,7.5)*+{\mathcal{X}'}="SX";
    (0,-7.5)*+{\mathcal{X}}="X";
    (30,7.5)*+{\states(\mathcal{H})'}="SH";
    (30,-7.5)*+{\states(\mathcal{H})}="H";
    {\ar@{|->}"SH";"H"};
    {\ar@{|->}"SX";"X"};
    {\ar"X";"H"_{\rho}};
    {\ar@{-->}"SX";"SH"^{\rho'}};
\endxy
\]

Since the functor $F$ is faithful, and rarely full, the space of all encodings $\rho$ in $\mathbf{Set}$ is reduced to a smaller space of structure-preserving encodings $\rho'$ in $\mathbf{C}$. In this way, searching for those encodings $\rho$ that lift to such a $\rho'$ might provide a method by which to simplify the search space when designing quantum encodings that are adapted to the problem being considered.

This seemingly abstract design goal is made to be compatible with our earlier examples. 

\begin{example}
Let $\mathcal{X}':=(\mathcal{X},\alpha)$ be a data domain that is a $G$-set, i.e., an object of $G\text{-}\mathbf{Set}$, which means that the data domain is equipped with a symmetry coming from a group action. Also, if $G\ni g\mapsto V_g\in\unitary(\mathcal{H})$ is a unitary representation of a group $G$ on the Hilbert space $\mathcal{H}$, with the induced action of $G$ on $\states(\mathcal{H})$ via $G\ni g\mapsto \Ad_{V_g}$, then $\states(\mathcal{H})':=(\states(\mathcal{H}),\Ad_{V})$ is an object of $G\text{-}\mathbf{Set}$. If $\rho:\mathcal{X}\to\states(\mathcal{H})$ is a quantum state encoding that is $G$-equivariant, which means that the symmetry is reflected in the quantum encoding, then the quantum encoding ``lifts'' to a morphism $\rho':\mathcal{X}'\to \states(\mathcal{H})'$ in $G\text{-}\mathbf{Set}$.
In this case, the general diagram we drew earlier becomes 
\[
\xy 0;/r.25pc/:
    (-30,7.5)*+{G\text{-}\mathbf{Set}}="SS";
    (-30,-7.5)*+{\mathbf{Set}}="S";
    {\ar"SS";"S"_{F}};
    (0,7.5)*+{(\mathcal{X},\alpha)}="SX";
    (0,-7.5)*+{\mathcal{X}}="X";
    (30,7.5)*+{(\states(\mathcal{H}),\Ad_{V})}="SH";
    (30,-7.5)*+{\states(\mathcal{H})}="H";
    {\ar@{|->}"SH";"H"};
    {\ar@{|->}"SX";"X"};
    {\ar"X";"H"_(0.45){\rho}};
    {\ar@{-->}"SX";"SH"^(0.45){\rho}};
\endxy
\]
where $F:G\text{-}\mathbf{Set}\to\mathbf{Set}$ is the forgetful functor from Example~\ref{ex:forgetfulGset}. 
Note that the map $\rho$ in $G\text{-}\mathbf{Set}$ is the same $\rho$ as in $\mathbf{Set}$ when viewed as a function (hence the same notation), which is because $F$ is faithful.  
Thus, by restricting to the subset of those $\rho$ that lift from $\mathbf{Set}$ to $G\text{-}\mathbf{Set}$, we search within a smaller space of suitable embeddings of $G$-equivariant morphisms in $G\text{-}\mathbf{Set}$. 

\end{example}

Our categorical proposal extends far beyond the setting of geometric quantum machine learning. We have already described other examples involving topological structure, smooth structure, and metric structure. Additionally, there is a category $\mathbf{LieOPS}$, defined in Appendix~\ref{app:onepgroupqencode}, that identifies the structure preserved under quantum unitary encodings $U:\R^{d}\to\unitary(\mathcal{H})$ of the form 
\begin{equation}
\label{eq:unitaryLieencodingconclusion}
U(x)=e^{-i \mathcal{L}(x)},
\end{equation}
where $\mathcal{L}:\R^{d}\to\mathcal{B}(\mathcal{H})$ is a linear transformation sending each point $x\in\R^{d}$ to a Hermitian operator $\mathcal{L}(x)$ acting on a Hilbert space $\mathcal{H}$ (here, $\mathcal{B}(\mathcal{H})$ denotes the set of bounded operators on $\mathcal{H}$). Such encodings are common in the quantum machine learning literature, and we have already seen several examples in the context of angle encoding and in geometric quantum machine learning. As proven in Appendix~\ref{app:onepgroupqencode}, such morphisms are obtained by first equipping sets with a structure related to one-parameter subgroups. In essence, the morphisms in $\mathbf{LieOPS}$ preserve these one-parameter subgroups, and Theorem~\ref{thm:LieOPSmorphism} proves that all unitary encodings $U:\R^{d}\to\unitary(\mathcal{H})$ that are morphisms in $\mathbf{LieOPS}$ (i.e., that preserve one-parameter subgroups) are of the form~\eqref{eq:unitaryLieencodingconclusion}. This structure reduces the space of quantum encodings to a \emph{finite-dimensional} space. Combining this with other structures allows one to further reduce the space of encodings relevant in a quantum machine learning task. Indeed, we have already seen an example worked out in Section~\ref{sec:symmetry} in the setting of geometric quantum machine learning, where the preservation of one-parameter subgroups together with equivariance led to a simple description of Lie algebra generators allowed for equivariant encodings.

\section{Discussion, open questions, and outlooks}
\label{sec:OQMF}

As quantum machine learning is a relatively young subject, there are many open questions to answer and possibly even more questions that are waiting to be formalized with the proper mathematical tools. As an example, a frequently asked question is one of quantum advantage. For instance, Bowles, Ahmed, and Schuld, through a number of experiments, argued that linearly separable data does not lend itself to quantum advantage~\cite{bowles2024}. As another example, the recent developed method of dequantization shows that certain classes of quantum algorithms perform, at best, just as well as some classical algorithms~\cite{Tang2022dequantizing,Tang2019,Tang2021,cotler2021revisiting,shindequantizing2024}. Based on these and many other examples, the current understanding is that structure in a dataset and/or task is important in determining possible quantum advantage~\cite{AaAm14,larocca2022group,thanasilp2024exponential,bowles2024,huang2021power}. In order to make progress in this direction, it therefore seems necessary to mathematically formalize the meaning of structure for broad applicability. 

As a step in this direction, we have argued that category theory offers a useful toolbox for organizing and capturing structure preservation for quantum encodings in the context of quantum machine learning. Different data sets and tasks admit structure that are relevant to specific problems, and identifying the appropriate category of such structure enables one to isolate a subset of quantum encodings preserving that structure. A prime example is geometric quantum machine learning~\cite{ragone2023representation,meyer2023exploiting,Nguyen_equivariantQNN2024,larocca2022group}, where the relevant structure was isolated as coming from the specific category of group actions on sets. We illustrated our perspective through several other examples to highlight that geometric quantum machine learning forms a single instance of our broader categorical formalism. Since finding a mathematical framework to isolate what we mean by structure in quantum machine learning tasks has been a long-standing problem~\cite{bowles2024,larocca2022group,SchuldKilloranAdvantage2022}, we hope that our perspective may offer insight towards extending the successes of geometric quantum machine learning to situations where other structures are preserved. 

There are many open questions and avenues of research that our perspective brings forth. 

First and foremost, we have only provided the basic ingredients for utilizing category theory in the context of QML at the level of the state encoding. A next step is to incorporate learning tasks, such as classification and/or regression. In particular, including full quantum machine learning models, which contain the quantum encoding as just one piece, should also be accounted for in the categories that preserve the structure. 

Second, now that structure preservation in the context of quantum encoding has been provided with a mathematical formalism, it is important to find the trade-offs between structure preservation, expressivity, generalization, performance, efficiency, classical simulability, etc.\ in this context~\cite{Anschuetz2023,cerezo2024doesprovableabsencebarren}. In particular, what QML algorithms admit a genuine quantum advantage on noisy intermediate scale quantum devices~\cite{PreskillNISQ2018}? Or should we aim to achieve a different goal~\cite{SchuldKilloranAdvantage2022}? As suggested in recent studies, a future direction for the QML community is to find data with structures that make the QML algorithms difficult to simulate classically~\cite{cerezo2024doesprovableabsencebarren}. We hope that the categorical perspective offered in this paper may help materialize such questions. 

\bigskip
\noindent
\textbf{Acknowledgements}

\noindent
AJP thanks Cheyne Glass, Seth Lloyd, and Alexander Schmidhuber for discussions.

\bigskip
\noindent
\textbf{Conflict of Interest Statement}

\noindent
AJP has received financial support from Deloitte in his involvement with this project.
He carried out this project as a consultant to Deloitte and not as part of his MIT responsibilities.

\bigskip
\noindent
\textbf{Disclaimer}

\noindent
About Deloitte: Deloitte refers to one or more of Deloitte Touche Tohmatsu Limited (``DTTL''), its global network of member firms, and their related entities (collectively, the ``Deloitte organization''). DTTL (also referred to as ``Deloitte Global'') and each of its member firms and related entities are legally separate and independent entities, which cannot obligate or bind each other in respect of third parties. DTTL and each DTTL member firm and related entity is liable only for its own acts and omissions, and not those of each other. DTTL does not provide services to clients. Please see \href{www.deloitte.com/about}{www.deloitte.com/about} to learn more.

Deloitte is a leading global provider of audit and assurance, consulting, financial advisory, risk advisory, tax and related services. Our global network of member firms and related entities in more than 150 countries and territories (collectively, the ``Deloitte organization'') serves four out of five Fortune Global 500\textregistered\;  companies. Learn how Deloitte's
approximately 460,000 people make an impact that matters at \href{www.deloitte.com}{www.deloitte.com}. 
This communication contains general information only, and none of Deloitte Touche Tohmatsu Limited (``DTTL''), its global network of member firms or their related entities (collectively, the ``Deloitte organization'') is, by means of this communication, rendering professional advice or services. Before making any decision or taking any action that
may affect your finances or your business, you should consult a qualified professional adviser. No representations, warranties or undertakings (express or implied) are given as to the accuracy or completeness of the information in this communication, and none of DTTL, its member firms, related entities, employees or agents shall be liable or
responsible for any loss or damage whatsoever arising directly or indirectly in connection with any person relying on this communication. 
Copyright \copyright\; 2024 Deloitte Development LLC. All rights reserved.

\appendix 

\section{One parameter groups for quantum encodings}
\label{app:onepgroupqencode}

Many of the unitary encodings $U:\mathcal{X}\to\unitary(\mathcal{H})$ considered in quantum machine learning are of the form 
\begin{equation}
\label{eq:unitaryLieencoding}
U(x)=e^{-i \mathcal{L}(x)},
\end{equation}
where $\mathcal{L}:\mathcal{X}\to\mathcal{B}(\mathcal{H})$ is a linear transformation sending each point $x\in\mathcal{X}$ (from a real vector space $\mathcal{X}$) to a Hermitian operator $\mathcal{L}(x)$ acting on a Hilbert space $\mathcal{H}$, where $\mathcal{B}(\mathcal{H})$ denotes the set of bounded operators on $\mathcal{H}$~\cite{meyer2023exploiting,SchuldPetruccione21,schuld2021effect}. One then obtains a quantum state encoding by fixing some fiducial state $\ket{\psi_0}\in\mathcal{H}$ via the formula $\rho(x)=U(x)\ket{\psi_0}\bra{\psi_0}U(x)^{\dag}$. Examples include the one from Section~\ref{sec:symmetry} on geometric quantum machine learning as well as angle encoding from Section~\ref{sec:top_smooth}. These types of examples of quantum encodings can also be viewed as preserving a certain structure. The type of structure preserved is a bit more technical than the types of structures we have considered in the body of the paper, so we include the discussion here in this appendix for completeness and because of the fact that unitary encodings of the form~\eqref{eq:unitaryLieencoding} are ubiquitous in the quantum machine learning community. We begin with a reminder of one-parameter subgroups of Lie groups (all Lie groups here will be matrix Lie groups for simplicity)~\cite{gilmore2006lie,bincer2013lie,hall2015lie,CMS95,ragone2023representation}.

\begin{definition}
Let $G$ be a Lie group. A continuous \define{one-parameter subgroup} of $G$ is a continuous group homomorphism $\gamma:\R\to G$, where $\R$ is viewed as a group under addition, i.e., 
\begin{enumerate}[i.]
\itemsep0pt
\item
$\gamma$ is a continuous function,
\item
$\gamma(0)=1$, and 
\item
$\gamma(s+t)=\gamma(s)\gamma(t)$ for all $s,t\in\R$. 
\end{enumerate}
\end{definition}

Recall that every one-parameter subgroup $\gamma:\R\to G$ of a Lie group $G$ is generated by a unique element $M\in\mathfrak{g}$ in the Lie algebra $\mathfrak{g}$ by the exponential map $\exp:\mathfrak{g}\to G$. Namely, 
\begin{equation}
\label{eqn:OPSexp}
\gamma(t)=\exp(t M)
\end{equation}
for all $t\in \R$. In particular, this implies $\gamma$ is not only continuous but smooth as well. Note that by differentiating the curve $\gamma$ at $t=0$ (cf.\ Section 1.3 in Ref.~\cite{BaMu94}), one obtains $M=\frac{d}{dt}\gamma\big|_{t=0}$. This example is precisely what is used when one has a data domain $\mathcal{X}$ of the form $\mathcal{X}=\R$, since in this case $x\mapsto e^{-ix L}$ with $L\in\mathfrak{su}(2^n)$ for a quantum unitary encoding mapping one-dimensional classical data onto $n$ qubits. It is also worth noting that the exponential map $\exp:\R^{d}\to\R^{d}$ from the Lie algebra of the \emph{additive} Lie group $\R^{d}$ to the additive Lie group $\R^{d}$ is the identity map, $\exp(x)=x$ for all $x\in\R^{d}$. 

However, for classical data that is provided in more than one dimension, we need to be careful about the ordering of operators obtained from such exponentials. To expound on this and to make more rigorous comparisons later, we recall the following result relating Lie group homomorphisms and Lie algebra homomorphisms. (For this theorem, recall that $[V,W]:=VW-WV$ denotes the commutator.)

\begin{theorem}
\label{thm:LieGrouptoAlgebra}
Let $G$ be any connected Lie group, let $H$ be any Lie group, and let $f:G\to H$ be a Lie group homomorphism. Then there exists a unique real-linear map $\mathcal{M}:\mathfrak{g}\to\mathfrak{h}$ such that 
\begin{equation}
\label{eqn:expnattrans}
f\big(\exp(W)\big)=\exp\big(\mathcal{M}(W)\big)
\end{equation}
for all $W\in\mathfrak{g}$. Moreover, $\mathcal{M}$ is a Lie algebra homomorphism, meaning that \begin{equation}
\big[\mathcal{M}(W_1),\mathcal{M}(W_2)\big]=\mathcal{M}\big([W_1,W_2]\big)
\end{equation}
for all $W_1,W_2\in\mathfrak{g}$. 
\end{theorem}

This theorem is well known (see Theorem 3.28 in Ref.~\cite{hall2015lie} for example), so we will not give a proof here because we will need a generalization later whose proof will be different. 
In contrast to the setting of Theorem~\ref{thm:LieGrouptoAlgebra}, the types of morphisms between Lie groups that describe quantum encodings of the form~\eqref{eq:unitaryLieencoding} are \emph{not} Lie group homomorphisms in general, primarily due to the lack of commutativity of arbitrary unitary gates in quantum circuits. 
A simple example illustrating this can be seen with a 2-dimensional data domain $\mathcal{X}=\R^2$. Suppose that $U:\mathcal{X}\to\unitary(\C^2)$ is a quantum unitary encoding of the form 
\begin{equation}
\label{eqn:unitaryencodingex}
U\big((s,t)\big)=e^{-is X-itY}
\end{equation}
for all $(s,t)\in\R^{2}$. 
Then this defines a perfectly reasonable quantum unitary encoding map that additionally satisfies~\eqref{eqn:expnattrans}, as will be explained momentarily. However, note that~\eqref{eqn:unitaryencodingex} does not define a group homomorphism because 
\begin{align}
U\big((s,0)\big)U\big((0,t)\big)&=e^{-is X}e^{-it Y}\nonumber\\
&\ne e^{-i(sX+tY)}\nonumber \\ 
&=U\big((s,0)+(0,t)\big)
\end{align}
in general. 

In more detail, we still have the following properties associated with~\eqref{eqn:unitaryencodingex}. First, let $\mathcal{M}:\R^{2}\to\mathfrak{su}(2)$ be the linear map uniquely determined by $\mathcal{M}\big((1,0)\big)=-iX$ and $\mathcal{M}\big((0,1)\big)=-i Y$, so that $\mathcal{M}\big((s,t)\big)=-isX-itY$ for all $(s,t)\in\R^{2}$. Since the exponential map $\exp:\R^{2}\to\R^{2}$ going from the Lie algebra of the Lie group $\R^{2}$ (with group structure given by addition) to the Lie group $\R^{2}$ coincides with the identity map, equation~\eqref{eqn:expnattrans} still holds and it agrees with~\eqref{eqn:unitaryencodingex}. Therefore, the assumption that $f:G\to H$ is a group homomorphism in Theorem~\ref{thm:LieGrouptoAlgebra} is not necessary for~\eqref{eqn:expnattrans} to hold. Second, for every one-dimensional additive subgroup $\R_{\alpha}\subset\R^{2}$ generated by some nonzero vector $\alpha$, i.e., 
\begin{equation}
\R_{\alpha}=\big\{r\alpha\;:\;\alpha\in\R^{2}\setminus\{0\},\,r\in\R\big\},
\end{equation}
the restriction of $U$ to this subgroup does indeed define a Lie group homomorphism $U\big|_{\R_{\alpha}}:\R_{\alpha}\to \unitary(\C^2)$ since 
\begin{align}
U(r\alpha)U(u\alpha)&=e^{-i rZ} e^{-iu Z}\nonumber \\
&=e^{-i(r+u)Z} \nonumber\\
&=U\big(r\alpha+u\alpha\big), \label{eq:Uraua}
\end{align}
where $Z:=\alpha\cdot(X,Y)$ is the dot product of $\alpha$ with the vector of operators $(X,Y)$. The second equality in~\eqref{eq:Uraua} holds because every matrix $Z$ commutes with itself. In other words, the one-parameter subgroup $\R\xrightarrow{\alpha}\R^{2}$ of $\R^2$ sending $r\in\R$ to $\alpha(r):=r\alpha$ is a one-parameter subgroup of $\R^2$ that gets pushed forward to a one-parameter subgroup 
\begin{equation}
\R\xrightarrow{\alpha}\R^{2}\xrightarrow{U}\unitary(\C^2)
\end{equation}
of $\unitary(\C^2)$ that sends $r\in\R$ to $U(r\alpha)$ \footnote{There is a slight abuse of notation here that is commonly done in category theory. The vector $\alpha\in\R^2$ is identified with the additive map $\alpha:\R\to\R^2$ sending $1\in\R$ to $\alpha$.}.

The previous discussion hints that the quantum unitary encodings~\eqref{eq:unitaryLieencoding} are more closely related to one-parameter subgroups when restricted to one-dimensional Lie subalgebras. 
We make this precise by introducing the category $\mathbf{LieOPS}$ of Lie groups and one-parameter subgroup homomorphisms, which makes use of the notion of a smooth map between smooth manifolds~\cite{BaMu94,LeeISM13}. 

\begin{definition}
Let $\mathbf{LieOPS}$ be the category (cf.\ Section~\ref{sec:catsbg}) whose objects are 
Lie groups and where a morphism $f:G\to H$ from a Lie group $G$ to a Lie group $H$ is a \define{one-parameter subgroup (OPS) homomorphism}, i.e., $f$ is a smooth map and for every one-parameter subgroup $\R\xrightarrow{\gamma}G$ in $G$, the composite $\R\xrightarrow{\gamma}G\xrightarrow{f}H$ is a one-parameter subgroup in $H$. 
\end{definition}

One can show that $\mathbf{LieOPS}$ is indeed a category. It is important to note that although every Lie group homomorphism is an OPS homomorphism, the converse is not true. Namely, not every OPS homomorphism is a Lie group homomorphism, with~\eqref{eqn:unitaryencodingex} providing an explicit example. The following theorem is a generalization of Theorem~\ref{thm:LieGrouptoAlgebra} and provides a characterization for how OPS homomorphisms can always be written in a form analogous to quantum unitary encodings as in~\eqref{eq:unitaryLieencoding}. 

\begin{theorem}
\label{thm:LieOPSmorphism}
Let $G\xrightarrow{f}H$ be an OPS homomorphism of Lie groups. 
Then there exists a unique linear map $\mathcal{M}:\mathfrak{g}\to\mathfrak{h}$ such that 
\begin{equation}
\label{eqn:OPGform}
f\big(\exp(x)\big)=\exp\left(\mathcal{M}(x)\right)
\end{equation}
for all $x\in\mathfrak{g}$. 
\end{theorem}

Our proof of Theorem~\ref{thm:LieOPSmorphism} will be different than the proof of Theorem~\ref{thm:LieGrouptoAlgebra} given in Ref.~\cite{hall2015lie} because we cannot assume that $f:G\to H$ is a group homomorphism. We assume some tools from differential geometry in the following proof~\cite{BaMu94,LeeISM13,MilnorTopDiff1997}.

\begin{proof} [Proof of Theorem~\ref{thm:LieOPSmorphism}]
First let $\mathcal{M}:\mathfrak{g}\to\mathfrak{h}$ be the \emph{differential} (sometimes called the \emph{pushforward}) of $f$ at the identity $1_{G}\in G$, i.e., $\mathcal{M}:=Df\big|_{1_{G}}$.
To be somewhat self-contained, let us recall how $\mathcal{M}$ is defined (the following definition is valid whenever $f$ is a smooth map between smooth manifolds, not necessarily Lie groups). For any smooth curve $\gamma:\R\to G$ with $\gamma(0)=1_{G}$ and $\frac{d}{dt}\gamma\big|_{t=0}=x\in\mathfrak{g}$, we have $\frac{d}{dt}(f\circ\gamma)\big|_{t=0}=\mathcal{M}(x)$. The fact that $\mathcal{M}$ defines a linear map follows from Exercise~17 in Section 1.3 of Ref.~\cite{BaMu94} or Proposition 3.6 of Ref.~\cite{LeeISM13}.

Now suppose that $\gamma$ is a one-parameter subgroup of $G$ with $\frac{d}{dt}\gamma\big|_{t=0}=x$. By~\eqref{eqn:OPSexp} for $\gamma$, this means that 
\begin{equation}
\label{eq:gammaexptx}
\gamma(t)=\exp(tx)
\end{equation}
for all $t\in\R$. 
Since $f$ is an OPS homomorphism, the composite $\R\xrightarrow{\gamma}G\xrightarrow{f}H$ is a one-parameter subgroup of $H$. Hence, by~\eqref{eqn:OPSexp} for $f\circ\gamma$, there exists a unique element $y\in\mathfrak{h}$ such that 
\begin{equation}
\label{eq:fgammaexpty}
(f\circ\gamma)(t)=\exp(ty)
\end{equation}
for all $t\in\R$. Putting these results together, we have
\begin{align}
y&=\frac{d}{dt}\exp(ty)\big|_{t=0} =\frac{d}{dt}(f\circ\gamma)\big|_{t=0}\quad\mbox{by~\eqref{eq:fgammaexpty}}\nonumber\\
&=Df\big|_{1_{G}}\left(\frac{d}{dt}\gamma\big|_{t=0}\right)\quad\mbox{by the chain rule}\nonumber\\
&=Df\big|_{1_{G}}\left(\frac{d}{dt}\exp(tx)\big|_{t=0}\right)\quad\mbox{by~\eqref{eq:gammaexptx}}\nonumber\\
&=\mathcal{M}(x)\quad\mbox{by definition of $\mathcal{M}$.}
\end{align}
Since $x$ was arbitrary, this together with applying $f$ to~\eqref{eq:gammaexptx} proves that 
\begin{equation}
f\big(\exp(tx)\big)=\exp\big(t\mathcal{M}(x)\big)
\end{equation}
for all $t\in\R$ and $x\in\mathfrak{g}$.
\end{proof}

We mention an immediate corollary relevant for quantum unitary encodings by taking $G=\R^d$ to be the additive group and $\mathfrak{g}=\R^{d}$ to be its associated Lie algebra. Note that in this case, $\exp:\mathfrak{g}\to G$ is the identity map as a function. 

\begin{corollary}
\label{cor:LieOPSmorphismforQencoding}
Let $\R^d\xrightarrow{f}G$ be an OPS homomorphism of Lie groups. 
Then there exists a unique linear map $\mathcal{M}:\R^d\to\mathfrak{g}$ such that 
\begin{equation}
\label{eqn:OPSformencoding}
f(x)=\exp\big(\mathcal{M}(x)\big)
\end{equation}
for all $x\in\R^d$. Conversely, given any linear map $\mathcal{M}:\mathfrak{g}\to\mathfrak{h}$, the function $f:\R^d\to G$ specified by~\eqref{eqn:OPSformencoding} defines an OPS homomorphism $\R^d\xrightarrow{f}G$. 
\end{corollary}

Corollary~\ref{cor:LieOPSmorphismforQencoding} says that any quantum unitary encoding $U:\R^{d}\to\unitary(\C^{n})$ of the form
\begin{equation}
U(x)=e^{-i\mathcal{L}(x)}, 
\end{equation}
where $\mathcal{L}:\R^{d}\to\mathcal{B}(\C^{n})$ is a linear map such that $\mathcal{L}(x)$ is self-adjoint for each $x\in\R^{d}$, 
defines a morphism $\R^{d}\xrightarrow{U}\unitary(\C^n)$ in $\mathbf{LieOPS}$. 
It is important to note that the unitary quantum encoding $U$ does \emph{not} define a Lie group homomorphism in general. This is because $U(x+y)=e^{-i\mathcal{L}(x+y)}$ is \emph{not} in general equal to $U(x) U(y)=e^{-i\mathcal{L}(x)}e^{-i\mathcal{L}(y)}$. Because the group structure as a whole is not preserved, this type of structure falls outside the context of geometric quantum machine learning. Instead, what is preserved is the group structure \emph{when restricted} to any one-dimensional subspace in $\R^d$. 

The benefit of Theorem~\ref{thm:LieOPSmorphism} in the context of quantum machine learning is that it offers a \emph{finite-dimensional} model of quantum encodings and shows what structures are preserved in the process. In fact, one can combine this structure together with symmetry to further reduce the space of quantum encodings. This is done in an illustrative example in Section~\ref{sec:symmetry}, with more details in Appendix~\ref{app:GQML}.

\section{Geometric Quantum Machine Learning}
\label{app:GQML}

In this appendix, we go into more details for Example~\ref{ex:Meyermodified} from Section~\ref{sec:symmetry}. We begin by describing how Figure~\ref{fig:meyersymmetrymodified} is generated. We then prove exactly how the space of \emph{equivariant} unitary quantum encodings is reduced to an 8-dimensional space. 

The example involves a learning task that distinguishes between two classes, labeled as $-1$ and $+1$, of points within a dataset $X\subseteq \mathbb{R}^2$. The labeling is determined by a binary \define{classifier}, which is a function $c:\mathcal{X}\to\{-1,0,+1\}$, where the $0$ element is included to allow for an undecided class. The set of all possible data $\mathcal{X}$ has a symmetry determined by the relations
\begin{equation}
\label{eqn:symmetryKlein4groupclassifierc}
c(x_1,x_2)=c(x_2,x_1)=c(-x_1,-x_2).
\end{equation}
Such a binary classifier \define{factors through the quantum encoding} $\rho:\mathcal{X}\to\states(\mathcal{H})$ if and only if there exists an observable $O\in\mathcal{B}(\mathcal{H})$ such that the spectrum of $O$ is $\{-1,+1\}$ and 
\begin{equation}
\label{eqn:classicalclassifier}
c(x)=\begin{cases}-1&\mbox{ if $y(x)<0$}\\0&\mbox{ if $y(x)=0$} \\ +1 &\mbox{ if $y(x)>0$}, \end{cases}
\end{equation}
where $y:\mathcal{X}\to\R$ is the function defined by 
\begin{equation}
\label{eqn:expectvaluexgeneral}
y(x)=
\Tr\big[\rho(x) O\big],
\end{equation}
which is the expectation value of $O$ in the state $\rho(x)$.
The functional $y$ is called the \define{quantum classifier} and the set $y^{-1}(0)\subseteq\mathcal{X}$ is called the \define{decision boundary} of the classifier~\footnote{The convex-linear functional sending states $\sigma\in\states(\mathcal{H})$ to $\Tr[\sigma O]$ defines a hyperplane $\{\sigma\in\states(\mathcal{H})\;:\; \Tr[\sigma O]=0\}$ in the space of states, which is one of the reasons it is often used for classification tasks in QML~\cite{schuld21}.}. The quantum classifier is \define{invariant} if and only if 
\begin{equation}
\label{eqn:symmetryKlein4groupclassifier}
y(x_1,x_2)=y(x_2,x_1)=y(-x_1,-x_2)
\end{equation}
for all $x_1,x_2\in\R$. Note that invariance of the quantum classifier~\eqref{eqn:expectvaluexgeneral} implies invariance of the classifier~\eqref{eqn:classicalclassifier}, i.e.,~\eqref{eqn:symmetryKlein4groupclassifierc} holds. 

Now let $\ket{\psi_{0}}\in\mathcal{H}=\C^{4}$ be the initial state
\begin{equation}
\label{eqn:psi0symmetrytoyexampleapp}
\ket{\psi_{0}}=\sqrt{p}\ket{+,+}-\sqrt{1-p}\ket{-,-}
\end{equation}
with $p=0.99$, let $O$ be the observable 
\begin{equation}
\label{eqn:observablesymmetryXX}
O=X\otimes X,
\end{equation}
and let $U:\mathcal{X}\to\unitary(\mathcal{H})$ be the quantum unitary encoding as in~\eqref{eqn:embeddingUsymmetryexample}. In this case, the quantum classifier can be computed explicitly as 
\begin{equation}
y(x_1,x_2)=\cos(x_1)\sin(x_2)+2\sqrt{p(1-p)}\sin(x_1)\sin(x_2).
\end{equation}
The orange region in Figure~\ref{fig:meyersymmetrymodified} consists of the points $(x_1,x_2)$ such that $y(x_1,x_2)<0$, while the blue region consists of the points $(x_1,x_2)$ such that $y(x_1,x_2)>0$. The boundary between the two regions consists of the points $(x_1,x_2)$ such that $y(x_1,x_2)=0$, which is the decision boundary.

Having shown how Figure~\ref{fig:meyersymmetrymodified} is obtained, we next analyze how equivariance reduces the space of possible quantum encodings. In Example~\ref{ex:Meyermodified}, we only looked at \emph{one} such encoding $U$ based on Equation~\eqref{eqn:embeddingUsymmetryexample}.  Meanwhile, a large set of unitary quantum encodings $U:\R^{2}\to\unitary(\C^{4})$ are of the form $U(x)= e^{-i \mathcal{L}(x)}$ for some linear transformation $\mathcal{L}:\R^{2}\to\mathfrak{su}(4)$ as in~\eqref{eqn:unitary4by4example}. Such a linear transformation $\mathcal{L}$ is uniquely determined by the value of $\mathcal{L}$ on the basis elements $e_{1},e_{2}$ of $\R^2$. Namely, $L_{1}:=\mathcal{L}(e_1)$ and $L_{2}:=\mathcal{L}(e_2)$. The fact that $L_{1}$ is a traceless Hermitian matrix says that $L_{1}$ is of the form
\begin{equation}
L_{1}=\begin{bmatrix}a_{11}&a_{12}&a_{13}&a_{14}\\\overline{a_{12}}&a_{22}&a_{23}&a_{24}\\\overline{a_{13}}&\overline{a_{23}}&a_{33}&a_{34}\\\overline{a_{14}}&\overline{a_{24}}&\overline{a_{34}}&a_{44}\end{bmatrix},
\end{equation}
where $a_{11},a_{22},a_{33},a_{44}\in\R$, $a_{12},a_{13},a_{14},a_{23},a_{24},a_{34}\in\C$, and $a_{11}+a_{22}+a_{33}+a_{44}=0$, and similarly for the matrix $L_{2}$. This parametrization leads to a $2(4+2(6)-1)=30$-dimensional vector space of unitary quantum encodings $U:\R^{2}\to\unitary(\C^{4})$ that are of the form $U(x)= e^{-i \mathcal{L}(x)}$ for some linear transformation $\mathcal{L}:\R^{2}\to\mathfrak{su}(4)$. 

The equivariance constraint $L_{2}=\SWAP L_{1} \SWAP$ from Equation~\eqref{eq:GMLexconstraint} says that $L_{2}$ is uniquely determined by $L_{1}$ thereby reducing the dimension of the space of unitary quantum encodings to $15$. Combining this with the equivariance constraint $\{X\otimes X,L_{1}\}=0$ from Equation~\eqref{eq:GMLexconstraint} yields
\begin{equation}
\begin{bmatrix}\overline{a_{14}}&\overline{a_{24}}&\overline{a_{34}}&a_{44}\\\overline{a_{13}}&\overline{a_{23}}&a_{33}&a_{34}\\\overline{a_{12}}&a_{22}&a_{23}&a_{24}\\a_{11}&a_{12}&a_{13}&a_{14}
\end{bmatrix}
=-\begin{bmatrix}a_{14}&a_{13}&a_{12}&a_{11}\\a_{24}&a_{23}&a_{22}&\overline{a_{12}}\\a_{34}&a_{33}&\overline{a_{23}}&\overline{a_{13}}\\a_{44}&\overline{a_{34}}&\overline{a_{24}}&\overline{a_{14}}\end{bmatrix}.
\end{equation}
Therefore, $L_{1}$ is of the form 
\begin{equation}
L_{1}=\begin{bmatrix}a_{11}&a_{12}&a_{13}&a_{14}\\\overline{a_{12}}&a_{22}&a_{23}&-\overline{a_{13}}\\\overline{a_{13}}&-a_{23}&-a_{22}&-\overline{a_{12}}\\-a_{14}&-a_{13}&-a_{12}&-a_{11}\end{bmatrix},
\end{equation}
where 
\begin{equation}
a_{12},a_{13}\in\C,\;\; 
a_{14},a_{23}\in i\R,\;\; 
a_{11},a_{22}\in\R.
\end{equation}
The set of such $L_{1}$ satisfying such conditions is an $8$-dimensional real vector space. One can explicitly check that the matrices 
\begin{equation}
\begin{aligned}[c]
&Z\otimes\mathds{1}_{2},\;\;  & \mathds{1}_{2}\otimes Z,\;\; \\[5pt]
&Y\otimes\mathds{1}_{2},
&\mathds{1}_{2}\otimes Y,\;\;
\end{aligned}
\qquad 
\begin{aligned}[c]
&X\otimes Z,\;\;
&Z\otimes X,\;\; \\[5pt]
&X\otimes Y,\;\; 
&Y\otimes X,\;\;
\end{aligned}
\end{equation}
are of this form. Since these matrices are linearly independent, this set of matrices forms a basis of this $8$-dimensional subspace.

\section{Classical linear metric learning and the Mahalanobis metric}
\label{app:CLMLMM}

Earlier, we saw that quantum metric learning can be described within a category of metric spaces and distance-preserving functions (embeddings). We can focus the discussion on metric spaces that have underlying sets $\R^{n}$ and the functions between them are linear. This leads to metrics known as Mahalanobis distance metrics. 

In more detail, let $\mathcal{X}=\R^{n}$ and equip it with a \define{Mahalanobis distance metric}, where $(\mathcal{Y}=\R^{m},d_{\mathcal{Y}})$ is a Euclidean space (so that $d_{\mathcal{Y}}$ is the Euclidean distance) and $\R^{n}\xrightarrow{\varphi}\R^{m}$ is an injective linear map, described by some $m\times n$ matrix $V$ with linearly independent columns~\cite{Kulis13,KaBi19}, in which case 
\begin{equation}
d_{\mathcal{X}}(x_1,x_2)=\sqrt{(x_1-x_2)^{T}(V^{T}V)(x_1-x_2)},
\end{equation}
where $V^{T}$ denotes the transpose of $V$.
Although this example is less relevant for \emph{quantum} metric learning, it is worth spending a few moments to illustrate the categorical structure associated with this in the context of \emph{classical} metric linear learning. 

Let $\mathbf{Mahal}$ be the category whose objects are pairs $(\R^{m},A)$, with $m\in\N$ and $A$ a positive-definite $m\times m$ matrix. A morphism in $\mathbf{Mahal}$ from $(\R^{m},A)$ to $(\R^{n},B)$ is an $n\times m$ matrix $V$ with trivial kernel such that $V^{T}BV=A$. The trivial kernel condition means that the nullspace of $A$ is $0$~\cite{Strang2022}, i.e., the matrix $A$ defines an injective linear transformation from $\R^{m}$ to $\R^{n}$. The category $\mathbf{Mahal}$ is \emph{equivalent} \footnote{An \emph{equivalence} of categories is a technical term whose formal definition we omit. Intuitively, however, two categories are said to be \emph{equivalent} if they are effectively the same.} to the subcategory of $\mathbf{Met}^{\mathrm{emb}}$ whose objects are Euclidean spaces equipped with the Mahalanobis distance metric and whose morphisms are required to be linear. To see this, first define the Mahalanobis metric on $\R^{m}$ associated with $A$ by 
\begin{equation}
d_{A}(x_1,x_2):=\sqrt{(x_1-x_2)^{T}A(x_1-x_2)}.
\end{equation}
Then the assumption $V^{T}BV=A$ shows that 
\begin{align}
d_{A}(x_1,x_2)&=\sqrt{(x_1-x_2)^{T}V^{T}BV(x_1-x_2)} \nonumber\\
&=\sqrt{(Vx_1-Vx_2)^{T}B(Vx_1-Vx_2)} \nonumber\\
&=d_{B}(Vx_1,Vx_2),
\end{align}
which proves that $V$ is distance-preserving. 
In this setting, if one is given data in $\R^{m}$ and the goal is to learn a Mahalanobis metric on $\R^{m}$ that captures the similarity in the data, then the task can be described  
as finding a linear map $\R^{m}\xrightarrow{V}\R^{n}$, since pulling back the Euclidean metric from $\R^{n}$ will give a Mahalanobis metric on $\R^{m}$. 

\section{Metric learning and semi-metric spaces}
\label{app:metriclearningsms}

In this appendix, we prove Lemma~\ref{lem:inducedmetric} and then provide a more detailed description of the quantum metric learning example shown in Figure~\ref{fig:semimetric1ddataset} in Section~\ref{ssec:dist_metriclearning}.

\begin{proof}[Proof of Lemma~\ref{lem:inducedmetric}]
Without assuming anything about $f$, 
\begin{align}
d(x_2,x_1)
&=d_{\mathcal{Y}}\big(f(x_2),f(x_1)\big) \nonumber\\
&=d_{\mathcal{Y}}\big(f(x_1),f(x_2)\big) \nonumber\\
&=d_{\mathcal{X}}(x_1,x_2)
\end{align}
shows that $d_{\mathcal{X}}$ is symmetric. Furthermore, 
\begin{align}
d_{\mathcal{X}}(x_1,x_2)&=d_{\mathcal{Y}}\big(f(x_1),f(x_2)\big) \nonumber \\
&\ge d_{\mathcal{Y}}\big(f(x_1),y\big)+d_{\mathcal{Y}}\big(y,f(x_2)\big)
\end{align}
holds for all $y\in \mathcal{Y}$. In particular, for any $x_{3}\in \mathcal{X}$, 
\begin{align}
d_{\mathcal{X}}(x_1,x_2)&\ge d_{\mathcal{Y}}\big(f(x_1),f(x_3)\big)+d_{\mathcal{Y}}\big(f(x_3),f(x_2)\big) \nonumber \\
&=d_{\mathcal{X}}(x_1,x_3)+d_{\mathcal{X}}(x_3,x_2),
\end{align}
which shows that $d_{\mathcal{X}}$ satisfies the triangle inequality. 
It is also true that $d_{\mathcal{X}}(x,x)=0$ for all $x\in \mathcal{X}$. 

Now, if $x_1,x_2\in \mathcal{X}$ satisfy
\begin{equation}
\label{eqn:metriczero}
0=d_{\mathcal{X}}(x_1,x_2)=d_{\mathcal{Y}}\big(f(x_1),f(x_2)\big), 
\end{equation}
then $f(x_1)=f(x_2)$. Hence, $f$ is one-to-one if and only if $d_{\mathcal{X}}(x_1,x_2)=0$ implies $x_1=x_2$ for all $x_1,x_2$ that satisfy~\eqref{eqn:metriczero}. 
\end{proof}

The goal, first outlined in Ref.~\cite{lloyd20}, is to classify a certain collection of data points on the interval $\mathcal{X}=[-2,2]$ into two classes. Although the given data $X\subset\mathcal{X}$ shown in Figure~\ref{fig:semimetric1ddataset} are provided as a subset of the Euclidean space $\R$, the data cannot have the induced metric from its embedding in Euclidean space because they are not linearly separable. Therefore, in order to arrive at a metric that more accurately reflects the similarity between points from one class and dissimilarity between points from different classes, we may apply a variation of techniques used in nonlinear metric learning~\cite{Bellet2015metric}. Ref.~\cite{lloyd20} takes the perspective to map these data points into the space of quantum states equipped with the Hilbert--Schmidt metric with the hope of more accurately reflecting the appropriate classification. 

Let $\{A,B\}$ be a partition of $X$, where the elements of $A$ are labeled as class $+1$ and the elements of $B$ are labeled as class $-1$. The associated classifier $c:\mathcal{X}\to\{-1,0,1\}$ satisfies 
\begin{equation}
\label{eqn:classifiermetriclearningex}
c(x)=\begin{cases}+1&\mbox{ if $x\in A$}\\ -1 &\mbox{ if $x\in B$}.\end{cases}
\end{equation}
Now, since the space of all states on $\mathcal{H}$ is a convex space, set
\begin{equation}
\label{eqn:rhoArhoB}
\rho_{A}:=\frac{1}{\# A}\sum_{a\in A}\rho(a)
\quad\text{ and }\quad
\rho_{B}:=\frac{1}{\# B}\sum_{b\in B}\rho(b),
\end{equation}
where $\#A$ and $\#B$ denote the cardinalities of $A$ and $B$, respectively. These density matrices are taken to be the empirical centroid density matrices for our classification task. In this case, we wish to \emph{maximize} the Hilbert--Schmidt distance between these two density matrices (cf.\ Definition~\ref{defn:quantumdistances}). Equivalently, we wish to \emph{minimize} the cost function
\begin{equation}
\label{eq:costfunctionLloyd}
C(\rho_{A},\rho_{B})=1-\frac{1}{2}d_{\mathrm{HS}}(\rho_A,\rho_B).
\end{equation} 
Note that finding an embedding that maximizes this distance is not a trivial task because the space of encodings $\mathcal{X}\to\states(\mathcal{H})$ is infinite-dimensional (in the sense that it defines an infinite-dimensional smooth space of paths~\cite{AdamsLoopspaces78,Mu00,PrSe1986,Br93,Pa15}). To make this maximization procedure more approachable, we look at a \emph{parameterized} family of encodings  
$\rho:\Theta\times\mathcal{X}\to\states(\mathcal{H})$, which are referred to as \emph{data re-uploading models}~\cite{PerezSalinas2020datareuploading,schuld2021effect,Jerbi2023}. (We will not discuss data re-uploading models in more detail in this work.) More specifically, we take $\Theta=[0,2\pi)^{4}$ and write $\theta=(\theta_1,\theta_2,\theta_3,\theta_4)\in\Theta$ and $\rho_{\theta}(x):=\rho(\theta,x)$. 
Moreover, we factor $\rho$ through a parameterized unitary quantum encoding $U:\Theta\times\mathcal{X}\to\unitary(\mathcal{H})$, where 
\begin{equation}
U_{\theta}(x)=\left(\prod_{n=1}^{4}
R_{X}(x)R_{Y}(\theta_n)\right)R_{X}(x),
\end{equation}
which is depicted in Figure~\ref{fig:semimetric1ddataset} and 
where our ordering convention on products of operators is 
\begin{equation}
\prod_{n=1}^{N}A_{n}=A_{N}\cdots A_{1}.
\end{equation}
Combining the unitary encoding with the ground state gives the quantum state encoding
\begin{equation}
\label{eq:rhothetaxLloyd}
\rho_{\theta}(x)=\ket{x}\bra{x},
\end{equation}
where 
\begin{equation}
\ket{x}=U_{\theta}(x) \ket{0}.
\end{equation}
Set 
\begin{equation}
O:=\rho_{A}-\rho_{B},
\end{equation}
which is an observable that depends on $\theta$ (since $\rho_A$ and $\rho_B$ depend on $\theta$), and define the quantum classifier $y:\mathcal{X}\to\R$ to be
\begin{equation}
\label{eq:quantummodelLloyd}
y(x)=\Tr\big[\rho(x)O\big].
\end{equation}

\begin{lemma}
\label{lem:oppositedensitymatrices}
Let $\rho_A$ and $\rho_B$ be two distinct qubit density matrices. 
Then there exists a $\lambda\in(0,1]$ such that the eigenvalues of $\rho_{A}-\rho_{B}$ are $\{-\lambda,+\lambda\}$. 
\end{lemma}

\begin{proof}
Since $\rho_{A}$ and $\rho_{B}$ are $2\times2$ self-adjoint matrices, let $\lambda_{1}$ and $\lambda_{2}$ denote the two eigenvalues of $\rho_{A}-\rho_{B}$. 
Since $\rho_{A}\ne\rho_{B}$, at least one of the two eigenvalues cannot be $0$ because $\rho_{A}-\rho_{B}$ is self-adjoint. 
Since 
\begin{equation}
\Tr[\rho_{A}-\rho_{B}]=\Tr[\rho_{A}]-\Tr[\rho_{B}]=1-1=0,
\end{equation}
the sum of the two eigenvalues of $\rho_{A}-\rho_{B}$ must vanish, i.e., 
$\lambda_{1}+\lambda_{2}=0$. In other words, $\lambda_{2}=-\lambda_{1}$. Setting $\lambda=|\lambda_{1}|$ proves the claim. 
\end{proof}

\begin{proposition}
\label{prop:Lloydclassifier}
Let $X\subset\mathcal{X}$ be a finite training data set together with a partition $A\cup B=X$ with elements of $A$ and $B$ labeled as class $1$ and $-1$, respectively. Let $\rho:\mathcal{X}\to\states(\mathcal{H})$ be any quantum encoding such that $\rho_A$ and $\rho_B$, as given by Eqn.~\eqref{eqn:rhoArhoB}, are distinct density matrices. Set $O=\rho_A-\rho_B$ and let $\{-\lambda,+\lambda\}$ denote the set of eigenvalues of $O$, with $\lambda>0$. Finally, let $c$ be the classifier in~\eqref{eqn:classifiermetriclearningex} and let $y:\mathcal{X}\to\R$ denote the associated quantum classifier as in~\eqref{eq:quantummodelLloyd}. Then
\begin{equation}
\mathrm{sgn}\big(y(x)\big)=c(x)
\end{equation}
for all $x\in X$, where $\mathrm{sgn}:(-\infty,0)\cup(0,\infty)\to\{-1,+1\}$ denotes the sign function $\mathrm{sgn}(x)=\frac{x}{|x|}$. 
\end{proposition}

Proposition~\ref{prop:Lloydclassifier} shows that the quantum classifier agrees with the \emph{fidelity classifier} used in Ref.~\cite{lloyd20} upon a rescaling. Meanwhile, the \emph{Helstr{\o}m classifier} of Ref.~\cite{lloyd20} uses the observable $\Pi_{+}-\Pi_{-}$, where $\Pi_{+}$ and $\Pi_{-}$ are the projections onto the positive and negative eigenspaces of $\rho_{A}-\rho_{B}$, respectively. 

Now, upon minimizing the cost function in~\eqref{eq:costfunctionLloyd}, we obtain specific values of $\theta_1,\theta_2,\theta_3,\theta_4$ for the parameterized encoding $\rho:\Theta\times\mathcal{X}\to\states(\C^{2})$. The associated encoding is not actually an injective map into the space of quantum states, as shown in Figure~\ref{fig:semimetric1ddataset}. This means that pulling back the metric gives a semi-metric, rather than a metric, on $\mathcal{X}$ (cf.\ Definition~\ref{defn:semimetricspace} and Lemma~\ref{lem:inducedsemimetric}). This semi-metric can be visualized in the same way as a metric, namely by illustrating the distance between points on a discretized grid from $\mathcal{X}\times \mathcal{X}$ (cf.\ Figure~\ref{fig:semimetric1ddataset} (c)). We note that the need for semi-metrics is not a phenomenon unique to quantum systems, as they already appear in linear metric learning, such as in dimensional reduction~\cite{Bellet2015metric}.

\section{Additional examples of functors and natural transformations}
\label{sec:nattransf}

This section contains an additional excursion into category theory that is relevant to some of the ideas touched upon in the main text. The mathematical details below lie somewhat outside of the summarized perspective in Section \ref{sec:lifting}, but we include the discussion here for the interested reader. And key to this discussion is the concept of a \emph{natural transformation}, which is essentially a mapping between functors. Historically, natural transformations played a pivotal role in the birth of category theory, which originally arose in the context of homology theory in algebraic topology~\cite{mac2013categories}.  
Although we did not make explicit use of natural transformations in the  main body of this work, it is interesting to point out that bit encoding can be described as a natural transformation that is closely related to a certain forgetful functor from $\mathbf{Vect}$ to $\mathbf{Set}$, which will be described momentarily. Moreover, the notion of $G$-equivariance can also be described as a natural transformation. This appendix is intended to go into details on these points in order to provide additional illustrations of categorical concepts through examples coming from quantum information theory and machine learning. In order to proceed, we will need to define two new concepts: natural transformations, of course, and that of composing functors. We present the latter first.

\begin{definition}
\label{defn:composefunctors}
Let $F:\mathbf{C}\to\mathbf{D}$ and $G:\mathbf{D}\to\mathbf{E}$ be two functors. The \define{composite} $G\circ F:\mathbf{C}\to\mathbf{E}$ is the functor that sends each object $\mathcal{X}$ in $\mathbf{C}$ to $G(F(\mathcal{X}))$ and sends each morphism $f:\mathcal{X}\to\mathcal{Y}$ in $\mathbf{C}$ to the morphism $G(F(f))$. 
\end{definition}

So, just as functions can be composed --- and more generally just as morphisms in a category can be composed --- functors can also be composed. This notion will appear in our discussion on bit encodings below. But first, we give the formal definition of a natural transformation.

\begin{definition}
Let $\mathbf{C}$ and $\mathbf{D}$ be categories, and let $F,G:\mathbf{C}\to\mathbf{D}$ be two functors. A \define{natural transformation} $\eta$ from $F$ to $G$, written $\eta:F\Rightarrow G$, associates to each object $\mathcal{X}$ in $\mathbf{C}$ a morphism $\eta_{\mathcal{X}}:F(\mathcal{X})\to G(\mathcal{X})$ in $\mathbf{D}$ such that the diagram
\begin{equation}
\label{eq:naturality}
\xy 0;/r.25pc/:
   (-12.5,7.5)*+{F(\mathcal{X})}="FX";
   (12.5,7.5)*+{F(\mathcal{Y})}="FY";
   (-12.5,-7.5)*+{G(\mathcal{X})}="GX";
   (12.5,-7.5)*+{G(\mathcal{Y})}="GY";
   {\ar"FX";"FY"^{F(f)}};
   {\ar"GX";"GY"_{G(f)}};
   {\ar"FX";"GX"_{\eta_{\mathcal{X}}}};
   {\ar"FY";"GY"^{\eta_{\mathcal{Y}}}};
\endxy
\end{equation}
in $\mathbf{D}$ commutes, i.e., $G(f)\circ\eta_{\mathcal{X}}=\eta_{\mathcal{Y}}\circ F(f)$, for every morphism $\mathcal{X}\xrightarrow{f}\mathcal{Y}$ in $\mathcal{C}$. The commutativity of diagram~\eqref{eq:naturality} is often referred to as \define{naturality}. 
\end{definition} 

The level of abstraction in this definition warrants several examples, and we focus on some relevant examples from quantum machine learning and quantum algorithms. Our first example arises from bit encoding, which is typically the first and most familiar encoding learned when studying quantum algorithms~\cite{DeutschJozsa92,NiCh11,BCRS19}. This example will involve both the concepts of  functor composition and a natural transformation. 

\begin{example}[Bit encoding]
\label{ex:bitencodingA}
There is a category $\mathbf{Vect}$ whose objects are complex vector spaces and whose morphisms are $\C$-linear transformations. 
Moreover, there is a familiar passage from $\mathbf{Set}$ to $\mathbf{Vect}$ that firstly associates to a set $\mathcal{X}$ the vector space $(\C_{\mathrm{fs}}^\mathcal{X},0,+,\cdot)$ generated by it, which will be defined in the next paragraph.
Secondly, this familiar passage also associates to a set-theoretic function $f:\mathcal{X}\to\mathcal{Y}$ a linear transformation $\C_{\mathrm{fs}}^{\mathcal{X}}\to\C_{\mathrm{fs}}^{\mathcal{Y}}$. This passage defines a functor $G:\mathbf{Set}\to\mathbf{Vect}$, and we will first describe its two ingredients one at a time. Afterwards, we will construct the natural transformation of bit encoding, which involves both $G$ and the forgetful functor $F:\mathbf{Vect}\to\mathbf{Set}$.

Mathematically, the vector space $(\C_{\mathrm{fs}}^\mathcal{X},0,+,\cdot)$ is the vector space of \emph{finitely-supported} complex-valued functions on $\mathcal{X}$. Namely, each vector $v$ in $\C_{\mathrm{fs}}^\mathcal{X}$ is by definition a function $v:\mathcal{X}\to\C$ such that $v(x)=0$ for all but finitely many $x\in\mathcal{X}$. (If the set $\mathcal{X}$ is finite, then every function will automatically be finitely-supported.) The zero vector $0$ is the function that assigns $0$ to every $x\in\mathcal{X}$. The sum of two functions $v,w\in\C_{\mathrm{fs}}^{\mathcal{X}}$ is the function $v+w$ whose value on $x\in\mathcal{X}$ is $(v+w)(x):=v(x)+w(x)$. Meanwhile, if $v$ is such a function and $\lambda\in\C$ is a scalar, then $\lambda \cdot v$, written as $\lambda v$ for short, is the function whose value on $x\in\mathcal{X}$ is $\lambda\, v(x)$.

The elements $x\in\mathcal{X}$ define a basis $\delta_{x}$ for $\C_{\mathrm{fs}}^{\mathcal{X}}$, where the function $\delta_{x}$ acts by $\delta_{x}(x')=\delta_{xx'}$ which is $1$ if $x'=x$ and $0$ if $x'\ne x$. 
If we were to use Dirac notation to express this vector, we could write $\ket{x}$ so that $\langle x' | x\rangle=\delta_{xx'}$, but that might be considered slightly abusive since we have not given $\C_{\mathrm{fs}}^{\mathcal{X}}$ the structure of a Hilbert space (although one can be defined in terms of our basis~\cite{Halmos1958}).
Moreover, we occasionally use the notation $\C_{\mathrm{fs}}^\mathcal{X}$ to refer to the vector space $(\C_{\mathrm{fs}}^\mathcal{X},0,+,\cdot)$, even though we should remember that the structure of a zero vector, vector addition, and scalar multiplication are implicitly assumed when referring to $\C_{\mathrm{fs}}^{\mathcal{X}}$ as a vector space. 

It is useful to view this from another perspective more familiar to the quantum information theorist. 
Let $\mathcal{X}=\Z_{2}^{n}=\{0,1\}^{n}$ be the set of all arrays $x=(x_0,x_1,\dots,x_{n-1})$ whose entries are either $0$ or $1$, so that we can think of $x$ as a binary representation of a number $\{0,1,2,\dots,2^{n}-1\}$. Then, the vector space  $\C^\mathcal{X}$ generated by $\mathcal{X}$ is precisely (i.e., naturally isomorphic to) 
\begin{equation}
\overbrace{\C^{2}\otimes\cdots\otimes\C^{2}}^{\text{$n$ times}}\cong\C^{2^n}, 
\end{equation}
which is the complex vector space of $n$ qubits. (There is no need to include the subscript $\mathrm{fs}$ in $\C_{\mathrm{fs}}^{\mathcal{X}}$ in this case because $\mathcal{X}$ is a finite set with $2^n$ elements.) 
From this perspective, each element $x=(x_0,x_1,\dots,x_{n-1})$ is associated with the vector 
\begin{equation}
\ket{x_0 x_1 \dots x_{n-1}}=\ket{x_0}\otimes\ket{x_1}\otimes\cdots\otimes\ket{x_{n-1}}.
\end{equation} 
This is what \emph{bit encoding} achieves (technically, bit encoding is more accurately described as a natural transformation, and we will make this more precise soon). 

Having described how the passage from $\mathbf{Set}$ to $\mathbf{Vect}$ associates to each set a vector space, let us next describe how we also obtain linear transformations from set-theoretic functions. Let $f:\mathcal{X}\to\mathcal{Y}$ be a function between sets. From this function, define the linear transformation $L_{f}:\C_{\mathrm{fs}}^{\mathcal{X}}\to\C_{\mathrm{fs}}^{\mathcal{Y}}$ as follows. By writing the basis elements of $\C_{\mathrm{fs}}^{\mathcal{X}}$ and $\C_{\mathrm{fs}}^{\mathcal{Y}}$ as $\ket{x}$ and $\ket{y}$ respectively, the linear transformation $L_{f}$ sends each basis vector $\ket{x}$ to $\ket{f(x)}$, which is another basis vector in $\C^{\mathcal{Y}}$. Since a linear transformation is uniquely determined by its action on a basis~\cite{Strang2022}, this defines the linear transformation $L_{f}$ on all of $\C_{\mathrm{fs}}^{\mathcal{X}}$. This technique of transforming $f:\mathcal{X}\to\mathcal{Y}$ to $L_{f}:\C_{\mathrm{fs}}^{\mathcal{X}}\to\C_{\mathrm{fs}}^{\mathcal{Y}}$ is often used in constructing quantum algorithms such as the Deutsch--Jozsa algorithm, which tests whether a function $f$ is constant or balanced~\cite{DeutschJozsa92,NiCh11,BCRS19}. 

So far, we have described a functor $G:\mathbf{Set}\to\mathbf{Vect}$ which takes a set $\mathcal{X}$ and constructs the vector space $(\C_{\mathrm{fs}}^{\mathcal{X}},0,+,\cdot)$ of complex-valued functions on $\C$. Now, there is \textit{also} a functor going the other way that takes a vector space $(V,0,+,\cdot)$ and forgets the vector space structure, leaving only the underlying set $V$. In more detail, there is  a functor $\mathbf{Vect}\to\mathbf{Set}$ that assigns to each vector space $V$ (technically, $(V,0,+,\cdot)$, where $0$ is the zero vector, $+$ is addition, and $\cdot$ is scalar multiplication) its underlying set $V$ of vectors, thus forgetting the linear structure (namely, the $0$ vector, vector addition $+$, and scalar multiplication $\cdot$). Moreover, this functor assigns every linear transformation to itself, since every linear transformation is a function. 

From these two functors, we can compose them (cf.\ Definition~\ref{defn:composefunctors}) in either order. One of these two ways of composing the functors will be used to exhibit bit encoding. Namely, taking a set $\mathcal{X}$, constructing the associated vector space $F(\mathcal{X})=(\C_{\mathrm{fs}}^{\mathcal{X}},0,+,\cdot)$, and then forgetting the vector space structure gives a \emph{set} $G(F(\mathcal{X}))=\C_{\mathrm{fs}}^{\mathcal{X}}$. 

We now have \emph{two} functors from $\mathbf{Set}$ to $\mathbf{Set}$. On the one hand, we have $G\circ F$, which we just described. On the other hand, we also have the functor $\id_{\mathbf{Set}}:\mathbf{Set}\to\mathbf{Set}$ that acts as the identity on all objects and morphisms. We can now realize bit encoding as the natural transformation $\eta:\id_{\mathbf{Set}}\Rightarrow G\circ F$ defined by sending a set $\mathcal{X}$ to a particular function $\eta_{\mathcal{X}}:\id_{\mathbf{Set}}(\mathcal{X})\to G(F(\mathcal{X}))$, which is a function of the form $\eta_{\mathcal{X}}:\mathcal{X}\to\C_{\mathrm{fs}}^{\mathcal{X}}$. The definition of this function is precisely the standard bit encoding, which sends an element $x\in\mathcal{X}$ to the vector $\eta_{\mathcal{X}}(x)=\ket{x}\in\C^{\mathcal{X}}$. So what does the naturality property~\eqref{eq:naturality} tell us about bit encoding? If $f:\mathcal{X}\to\mathcal{Y}$ is any (classical) function between sets, let $L_f:\C^{\mathcal{X}}\to\C^{\mathcal{Y}}$ denote associated linear transformation determined uniquely by how it acts on the basis $\{\ket{x}\}$. Naturality of $\eta$ then says that $\eta_{\mathcal{Y}}(f(x))=L_f(\eta_{\mathcal{X}}(x))$, which reads $\ket{f(x)}=L_{f}\ket{x}$, a completely \emph{natural} condition used in most quantum algorithms, and one which we in fact already used when defining the linear transformation $L_{f}$ by how it acts on a basis.
\end{example}

As our next example, one might have suspected that the very definition of an equivariant map as in Equation~\eqref{eq:Gsetmap} looks a lot like naturality. This is indeed the case after one views $G$-sets as functors and is analogous to how representations can be viewed as functors~\cite{Pa18a}. 

\begin{example}[$G$-equivariant maps revisited]
Let $G$ be a group and let $\mathbb{B}G$ be the category which consists of only a single object, denote it by $\bullet$. Moreover, the set of morphisms from $\bullet$ to itself in the category $\mathbb{B}G$ is defined to be the set $G$. The composition in $\mathbb{B}G$ is then taken to be the multiplication operation in $G$. One can check that $\mathbb{B}G$ is a category with these definitions. In this way, a group can be viewed as a category. A $G$-set can then be viewed as a functor $\alpha:\mathbb{B}G\to\mathbf{Set}$. Indeed, the unique object $\bullet$ in $\mathbb{B}G$ gets sent to some set, call it $\mathcal{X}$. Moreover, since $\alpha$ is a functor, it sends a morphism in $\mathbb{B}G$, which is an element $g$ of the group $G$, to some function $\alpha_{g}:\mathcal{X}\to\mathcal{X}$. Now, since the element $g$ is invertible, there exists a $g^{-1}$ such that $gg^{-1}=1_{G}=g^{-1}g$, where $1_{G}$ is the unit element of $G$. Hence, by functoriality of $\alpha$, we have that $\id_{\mathcal{X}}=\alpha_{e}=\alpha_{gg^{-1}}=\alpha_{g}\circ\alpha_{g^{-1}}$, and similarly $\id_{\mathcal{X}}=\alpha_{g^{-1}}\circ\alpha_{g}$, which proves that $\alpha_{g}$ is a bijection with inverse $\alpha_{g^{-1}}$. 

Now, consider \emph{two} such $G$-sets, viewed as functors $\alpha,\beta:\mathbb{B}G\to\mathbf{Set}$, where $\mathcal{X}:=\alpha(\bullet)$ and $\mathcal{Y}:=\beta(\bullet)$. A natural transformation $f:\alpha\Rightarrow\beta$ assigns to the only object $\bullet$ of $\mathbb{B}G$ a function $f:\mathcal{X}\to\mathcal{Y}$ (abusively denoted by the same letter). Naturality of $f$ as a natural transformation then gives exactly the $G$-equivariant condition~\eqref{eq:Gsetmap}. 
\end{example}

\bibliography{refs}
\end{document}